\documentclass[journal]{new-aiaa}

\usepackage[utf8]{inputenc}
\usepackage{amsmath}
\usepackage[version=4]{mhchem}
\usepackage{siunitx}
\usepackage{longtable,tabularx}
\usepackage[]{algorithm2e}
\usepackage{graphicx}

\usepackage{amsthm}

\usepackage{bm}
\usepackage{subfigure}

\DeclareMathOperator*{\argmin}{arg\,min}
\DeclareMathOperator*{\argmax}{arg\,max}

\newtheorem{nclem}{Lemma}
\newtheorem{nccon}{Condition}

\newtheorem{ncdef}{Definition}

\title{Subgoal Planning Algorithm for Autonomous Vehicle Guidance}
\author{
Andrew Feit
\enspace
and
\enspace
B\'{e}r\'{e}nice Mettler
\\
{ 
\normalsize
\itshape
iCueMotion, LLC. San Francisco, CA, 94103
}
}

\begin{document}
\maketitle
\begin{abstract}
Trained humans exhibit highly agile spatial skills, enabling them to operate vehicles with complex dynamics in demanding tasks and conditions. Prior work shows that humans achieve this performance by using strategies such as satisficing, learning hierarchical task structure, and using a library of motion primitive elements. A key aspect of efficient and versatile solutions is the decomposition of a task into a sequence of smaller tasks, represented by subgoals. The present work uses properties of constrained optimal control theory to define conditions that specify candidate subgoal states and enable this decomposition. The proposed subgoal algorithm uses graph search to determine a subgoal sequence that links a series of unconstrained motion guidance elements into a constrained solution trajectory. In simulation experiments, the subgoal guidance algorithm generates paths with higher performance and less computation time than an RRT* benchmark. Examples illustrate the robustness and versatility of this approach in multiple environment types.
\end{abstract}

\section*{Nomenclature}
{\renewcommand\arraystretch{1.0}
\noindent\begin{longtable*}{@{}l @{\quad=\quad} l@{}}
\multicolumn{2}{@{}l}{Sets} \\
$\mathcal{W}$  & Task workspace \\
$\mathcal{F}$ &  Free workspace \\
$\mathcal{O}_E$, $O_i$ & Environment obstacles \\
$\mathcal{Q}$ & Configuration space \\
$\Theta$ & Vehicle orientation space \\
$\mathcal{V}$ & Vehicle configuration rate space \\
$\mathcal{X}$ & System state space \\
$\mathcal{U}$ & Vehicle control space \\
$\mathcal{S}$ & Vehicle trajectory space \\
$\mathcal{G}$ & Subgoal candidate set \\
$\mathcal{D}$ & Guidance sub-domain \\
$\mathcal{B}$ & Constrained workspace domain \\
\multicolumn{2}{@{}l}{Functions} \\
$j$ & Differential cost function \\
$J$ & Trajectory cost functional\\
$H$ & Hamiltonian function\\
$f$ & State space dynamics\\
$k$ & Tracking control policy \\
$h$ & Perceptual information extraction function \\
$\Phi$ & Task transition function \\
$\pi$ & Guidance policy \\
$\Pi$ & Guidance mapping \\
$\gamma$ & Task planning map \\
$\mathcal{L}$ & Lagrange function \\
$c$ & Constraint function \\
$\Psi$ & Symmetry transformation \\
$\theta$ & Inverse cost function \\
$g$ & Bounding function \\
$V$ & Lyapunov function \\
$U$ & Utility function \\
\multicolumn{2}{@{}l}{Variables} \\
$t$, $t_0$, $T$ & Time, start time, final time \\
$\mathbf{x}$ & Vehicle state \\
$\mathbf{p}$ & Vehicle position \\
$\mathbf{v}$ & System velocity \\
$\mathbf{u}$ & Vehicle control \\
$\bar{s}$ & Trajectory \\
$g$ & Subgoal state \\
$\Gamma$ & Subgoal sequence \\
$\lambda$ & Hamiltonian costate-vector\\
$P$ & System plant \\
$\phi$ & Lagrange multiplier \\
$\mathbf{a}$, $\mathbf{b}$ & Example system states \\
$T_B$ & Separatrix \\
$S_B$ & Bounding trajectories \\
$\xi$ & Lyapunov integrating variable \\
$c$, $\alpha$ & Finite-time stability constants \\
$n$, $N$ & Number of subgoals \\
$\epsilon$ & Planner differential cost tolerance \\
$E$ & Subgoal graph edges \\
$H$ & Neighbor subgoal heuristic list \\
$\psi$ & Vehicle heading angle \\
$x$, $y$ & Vehicle position \\
$\theta_g$ & Bearing to subgoal \\
$k$ & Perceptual guidance constant \\
$d$ & Distance to subgoal \\
\multicolumn{2}{@{}l}{Subscripts} \\
$lat$, $lon$ & Lateral and longitudinal axes, e.g. steering and forward speed \\
$ref$ & Reference or desired value \\
$k$ & Subgoal index \\
\end{longtable*}}

\section{Introduction}
Humans and animals have a natural ability for motion guidance because it is a fundamental capability for interacting with the world. Human guidance encompasses motion of one's own body, as well as control of vehicles from a first-person perspective, such as driving a car or flying an aircraft. 
Observing motion behavior in competitive activities, such as athletics or vehicle control, shows that humans learn over time to integrate spatial planning (i.e. configuration space) with dynamic planning (i.e. velocities and forces) to achieve high task performance (such as in terms of travel time, speed, or effort). In addition to performance, humans excel in terms of adaptability to a wide range of environments, versatility to novel tasks not previously encountered, and robustness to uncertainty in system dynamics and constraints \cite{mettler2011structure}. Humans can also modulate their behavior to trade-off safety vs. performance in response to specific task requirements. 

While human brains as a whole can process large amounts of information, they have specific limitations. Humans have limited working memory, restricting the planning search space and the number of environment objects that can be considered simultaneously \cite{baddeley1992working}. In addition, humans are subject to problem solving limitations. For example, humans do not posses a shortcut to solve NP-complete combinatorial optimization problems. Finally, human senses, such as vision, proprioception, and inertial (by the vestibular system) provide only partial observability of the agent-environment state. Humans therefore must integrate their sensory inputs to achieve sufficient situational awareness \cite{borah1988optimal}. Despite these limitation, athletes and pilots achieve performance significantly beyond the capabilities of current autonomous or robotic systems. 

Mettler et al \cite{mettler2015systems} term the challenges above collectively as \textit{behavioral complexity}. Prior work describes two approaches by which humans reduce the behavioral complexity of motion planning tasks. First, Simon introduced the concept of satisficing \cite{simon1972theories} to describe behavior that may not be optimal with respect to a utility function, but satisfies constraints. 
This idea suggests that humans use a simplified planning process to obtain sub-optimal, constraint-satisfying solutions quickly and efficiently.
Simon identifies characteristics of tasks that make optimization difficult, such as uncertainty in the reward for taking specific actions, and uncertainty in the set of available actions at each instant in time. In addition, computation of the true cost or value function may be too complex. As a result, humans must employ heuristics, prune the decision tree, or approximate the true value function to make decisions efficiently. 
Second, Simon suggests that humans simplify a task by taking advantage of sparse agent-environment relationships. Simon introduces the idea that, similar to the study of other physical phenomenon, these relationships take the form of invariant quantities in agent-environment system dynamics, and must be identified to understand human behavior \cite{simon1990invariants}. 

Mettler and Kong extend the concept of invariants to motion control. The authors present experimental results indicating that humans exploit invariants to reduce problem complexity \cite{kong2013modeling}. From a task perspective, invariants take the form of equivalence classes across the task domain that decompose a problem into a smaller set of common subtasks. The identification of common subtasks is a form of structure learning \cite{braun2010structure}. Subtasks serve as a unit of organization that simplify the decision making domain. From a behavior perspective, invariants manifest as interaction patterns (IP) \cite{kong2009general, kong2013modeling, mettler2011structure, mettler2015systems, mettler2017emergent}. IP describe action-perception relationships that generate common units of behavior across multiple parts of a task-environment domain. The present work integrates these invariants of human behavior into an autonomous motion planner. The approach taken in this work accounts for the limitations of both biological and machine systems by applying concepts observed in human behavior to deal with task complexity.

\section{Problem Formulation and Overview}
\subsection{Research Objective}

In this paper we present a motion planning algorithm that uses subgoals to find a constrained optimal control solution for a vehicle navigating through an obstacle field. This problem is motivated by prior work that established a high-level description of human motion behavior \cite{mettler2015systems, mettler2017emergent}, including principles that define task representation in terms of subgoals. The specific objective of the present paper is to extend this description to formalize subgoal planning from the perspective of optimal control.

The planning algorithm we present uses two important concepts from prior work to define specific conditions needed to generate solutions efficiently, while accounting for computational and perceptual limitations. The first concept is that modeling human behavior as a hierarchical system accounts for a reduction in spatial control task complexity \cite{mettler2013hierarchical, mettler2015systems}. The hierarchical model formally delineates planning, guidance, and tracking functions. Within this structure, each function may be solved optimally, or sub-optimally in exchange for reduced computational requirements (i.e. satisficing \cite{simon1972theories}). The hierarchical model follows from the observation that humans use invariants in task-environment dynamics to reduce behavior representation requirements.
The second concept is that these invariants allow agents to generalize task knowledge as policies that are applicable across a task domain. Previous work investigated the role of invariants by identifying the guidance function that agents use to generate continuous motion trajectories between task states \cite{kong2013modeling, feit2015experimental, feit2016extraction}.  This research showed that segments of environment-agent behavior across a task domain are equivalent through a set of symmetry transformations, allowing a wide range of motion behavior to be described by a smaller set of unique interactions, or guidance primitives.

The work presented in this paper extends this subgoal identification and planning approach first presented by Feit and Mettler \cite{feit2015human} to formulate the trajectory optimization problem as a graph search across connected transition states, or subgoals (\cite{kong2009general, kong2013modeling, mettler2011structure, mettler2015systems}). This paper first derives properties of the constrained optimal control problem that provide rules for identifying these transition points as cost minima within the search domain. The paper then describes a planning algorithm that computes a solution as a sequence of subgoals. The agent deploys the resulting solution by generating continuous paths between subgoals using canonical guidance primitives.

This graph search approach is applicable to the case when the environment is fully-known to the agent, but provides the basis for future work needed to apply the concepts of this planning algorithm to exploratory motion guidance in partially known environments.

\subsection{Problem Formulation}
\begin{figure}[hbt]
\centering
\includegraphics[width=0.45\linewidth]{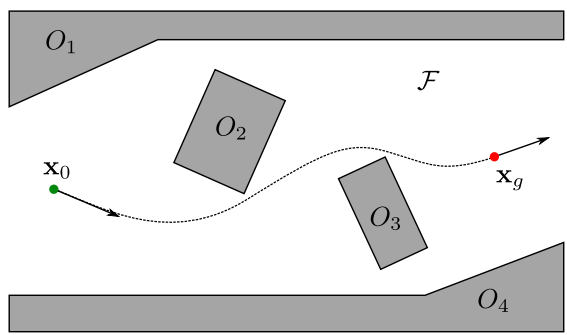}
\caption{Guidance problem overview.}
\label{fig:problem_overview}
\end{figure}
To integrate human motion guidance strategies into an autonomous guidance system, we first describe the task in terms of dynamics and control principles. This formulation provides a connection between human behavior characteristics, and dynamic trajectory optimization, which is the traditional framework for aerospace vehicle guidance. 

The guidance task is modeled as a constrained optimal control problem as depicted in Fig. \ref{fig:problem_overview}. The task consists of a state-space, which is the product space of the workspace, vehicle orientations, and system rates, $\mathcal{X} = \mathcal{W} \times \Theta \times \mathcal{V} \subset \mathbb{R}^n$, and a control space $\mathcal{U} \subset \mathbb{R}^m$. The workspace, $\mathcal{W} = \mathcal{O}_E \cup \mathcal{F}$, contains free space ($\mathcal{F}$) and environment objects ($\mathcal{O}_E$). Environment objects are a subset of the workspace $ \mathcal{O}_E \subset \mathcal{W} $ that act as constraints on feasible system trajectories. Constraints are defined by a set of inequalities, $ c_i(\mathbf{x}) \leq 0 $ for $ i = \left[ 1,n \right] $. As depicted in Fig. \ref{fig:problem_overview}, constraints can consist of a set of discrete, mutually disjoint objects, $\mathcal{O}_E = \bigcup_{i=1}^{k} O_i$, where each object may be an infinite subset, or a finite region. The task configuration space (or C-space) is $\mathcal{Q} = \mathcal{W} \times \Theta$ (and $\mathcal{Q}_{free} = \mathcal{F} \times \Theta$), where $\Theta$ describes system configurations such as vehicle component orientation that can change independently of workspace position. System rate $\mathcal{V}$ describes the rate of change of vehicle configuration. 
The task objective is to guide a system from an initial location, $\mathbf{p}_0 \in \mathcal{F}$ to a goal state, $ \mathbf{x}_g \in \mathcal{Q}_{free} \times \mathcal{V} $, while satisfying system dynamics and spatial constraints, and minimizing trajectory cost. 

Motion of the system through space is constrained by the dynamics of the vehicle or body, $\dot{\mathbf{x}}(t) = f(\mathbf{x}(t), \mathbf{u}(t))$, with control input sequence $\mathbf{u}(t) \in \mathcal{U}$, for time $t \in \left[ 0, T \right] \subset \mathbb{R} $, over a state-space domain, $\mathcal{X}$. 
The function $ \mathbf{x}(t) : \left[ 0, T \right] \rightarrow \mathcal{X} $ denotes a solution to the system beginning at state $\mathbf{x}_0 = \mathbf{x}(0)$ and ending at $\mathbf{x}_g = \mathbf{x}(T)$. 
The set $ \bar{s}_{u}(\mathbf{x}_0) = \lbrace \mathbf{x}(t) \, , \, \left[ 0, T \right] \rbrace $ contains all points on the trajectory. The cost of a trajectory $ \bar{s}_{u}(\mathbf{x}_0)$ resulting from control sequence $ \mathbf{u}(t) $ is defined in the form of the functional:
\begin{equation}
\label{eqn:objective}
J(\bar{s}) = \int_{0}^{T} j(\mathbf{x}(t),\mathbf{u}(t))\mathrm{d}t
\end{equation}
In Eqn. \ref{eqn:objective}, $j(\cdot)$ is the differential path cost. The guidance problem is that of determining a control sequence $\mathbf{u}(t)$ that drives the system from an initial state $\mathbf{x}_0$ to a goal state $\mathbf{x}_g$, and minimizes eqn. (\ref{eqn:objective}).
\subsection{Approach}
\subsubsection{Optimal Control}
The optimal control problem is posed as a minimization of the cost functional over the set of admissible control input series'. Analytical solutions are only possible in particular cases; including the linear quadratic regulator (LQR) feedback control and finite-time minimum energy feed-forward control \cite{athans2006optimal, hespanha2009linear}. Optimal control sequences throughout a system domain may be expressed as a policy function, $\mathbf{u}^* = k^*_{x_g}(\mathbf{x})$, with $J_{x_g}^*(\mathbf{x})$ expressing the optimal cost of a trajectory from state $\mathbf{x}$ to the goal $\mathbf{x}_g$ resulting from the optimal control sequence. While linear full-state control policies can be directly synthesized \cite{hespanha2009linear} for arbitrary state dimension, planning trajectories under nonlinear system dynamics and arbitrary constraints remains a hard problem. For example, a rotorcraft position controller must determine four control actuator values based on 14 or more states. To avoid this complexity, we use the following two concepts as introduced above: the hierarchical model of human guidance, and the spatial value function.

\subsubsection{Hierarchical Task Model}
The hierarchical model of human guidance (Fig. \ref{fig:hierarchical_model}) introduced by Mettler and Kong\cite{mettler2013hierarchical} partitions motion behavior into subgoal planning, kinematic guidance of a reference trajectory, and regulation of higher-order states. 
\begin{eqnarray}
\label{eqn:hierarchicsystemdynamics}
\text{Task transition:}&  g_{k+1} = \Phi(g_k, \pi_k) \nonumber \\
\text{Kinematics:}&  \dot{\mathbf{x}}_p(t) = \mathbf{v_{ref}}(t) \\
\text{Dynamics:}&  \dot{\mathbf{v}} = f(\mathbf{v},\mathbf{u}) \nonumber
\end{eqnarray}
Eqn. \ref{eqn:hierarchicsystemdynamics} delineates the system dynamics involved with each of these levels. To decouple the planning, guidance and tracking problems in this way, the models for each level must be separable; for example, the guidance level depends on the tracker to regulate the system to the reference velocity. Problem complexity is then reduced by identifying control relationships that are invariant across sub-domains of the task. For example, kinematics are invariant with respect to spatial translation and rotation of a target subgoal, so therefore every subgoal represents an equivalent sub-task. Similarly, tracking dynamics are invariant with respect to spatial location. 

\subsubsection{Spatial Value Function}
The separation of kinematics and system dynamics is enabled by the concept of a spatial value function (SVF). The SVF was introduced by Mettler and Kong to model human guidance behavior, based on precision interception experiments with miniature remote control helicopters \cite{mettler2013mapping}.
A SVF, also referred to as cost-to-go (CTG), specifies the optimal cost of a trajectory starting from a given system configuration $\mathbf{p} \in \mathcal{F}$ to the goal, denoted $J_{\mathbf{x}_g}(\mathbf{p})$. An optimal control policy based on the SVF defines reference system rates as a function of system configuration, i.e. $ \mathbf{v}_{ref} =\pi_{\mathbf{x}_g}(\mathbf{p}) : \mathcal{W} \times \mathcal{X} \rightarrow \mathcal{V} $. 
This policy, referred to as the velocity vector field (VVF), is the gradient field of the SVF, $\nabla J_{x_g}$. Optimal spatial VVF and CTG functions are: 
\begin{eqnarray}
\pi^*_{\mathbf{x}_g}(\mathbf{p}) &=& \argmin_{\mathbf{v}_{ref} \in \mathcal{V}} J_{\mathbf{x}_g}(\mathbf{p}, \mathbf{v}_{ref}) \\
J^*_{\mathbf{x}_g}(\mathbf{p}) &=& \min_{\mathbf{v}_{ref} \in \mathcal{V}} J_{\mathbf{x}_g}(\mathbf{p}, \mathbf{v}_{ref}) \nonumber
\label{eqn:vvf}
\end{eqnarray}
The VVF and CTG, acting in feedback with system dynamics, specify optimal configuration space trajectories. Integrating the VVF, i.e. $\dot{\mathbf{p}} = \pi_{\mathbf{x}_g}(\mathbf{p})$ from an initial configuration $\mathbf{p}_0$, generates a trajectory $ \bar{s}^{\pi}(\mathbf{p}_0,\mathbf{x}_g) $, with the cost $ J_{\mathbf{x}_g}^{\pi}(\mathbf{x}_0) = J(\bar{s}^{\pi}(\mathbf{x}_0,\mathbf{x}_g))$. As a result, the field $ \pi_{\mathbf{x}_g}(\mathbf{p}) $ defines a mapping $ \Pi_{\mathbf{x}_g} : \mathcal{W} \rightarrow \mathcal{S}_{x_{g}} $ from initial configurations to trajectories. Importantly, the resulting spatial policy is invariant with respect to symmetry transformations in the configuration space, allowing it to be re-used across a task domain as a planning-level behavior element. 

With this kinematic decoupling, the guidance task consists of learning a spatial guidance policy $\Pi$ that generates trajectories to a goal state. The tracking task consists of determining a control policy that stabilizes the vehicle velocity $\mathbf{v}(t)$ near the reference velocity $\mathbf{v}_{ref}(t)$ specified by the guidance level, for example using feedforward (such as by dynamic inversion) and feedback components \cite{mettler2013mapping}. The following work focuses on the guidance component and how it relates to subgoal planning.
\subsubsection{Human Guidance Policy}
Humans often do not generate guidance behavior that is optimal with respect to typical control costs such as settling time or energy \cite{simon1972theories}. Following the idea of bounded rationality, a sub-optimal solution can be faster or easier to generate for tasks in which computing an optimal policy is computationally intractable or when the perceptual signals needed to implement an optimal policy are not available. In addition, humans primarily rely on perceptual guidance, which means they are subject to perceptual information constraints. Prior work shows that humans can approximate optimal behavior using a perceptual guidance policy that maximizes the information channel between perceived visual cues and guidance actions \cite{tishby2011information}. To investigate the guidance strategies used by humans, prior work recorded experimental human guidance behavior in third and first-person navigation tasks \cite{mettler2013mapping, feit2015experimental}. The present research uses the vehicle dynamics and environment obstacle configuration used by Feit and Mettler in the human experiments \cite{feit2015experimental, feit2016extraction}. The resulting data shows that human behavior can be modeled by a consistent guidance policy, in the form of a SVF. Based on this idea, the present work defines a nominal policy, $ \pi_{\mathbf{x}_g}(\mathbf{p})$, that generates a sub-optimal reference trajectory given an initial configuration state and a goal state. The sub-optimal cost is greater than or equal to the cost of an optimal trajectory, $ J_{\mathbf{x}_g}^{\pi}(\mathbf{x}_0) = J(\bar{s}^{\pi}(\mathbf{x}_0,\mathbf{x}_g)) \ge J_{\mathbf{x}_g}^*(\mathbf{x}_0)$. 
\subsection{Overview}
The rest of the paper is organized as follows: Section (\ref{sec:related_work}) provides a background on related work in both computational methods for constrained optimal control and on current research in understanding human guidance behavior. Section (\ref{sec:theoretical_approach}) reviews the optimal control problem formulation. Section (\ref{sec:structural_properties}) applies properties of the constrained optimal control problem to define planning heuristics based on task-environment structure. Section (\ref{sec:implementation}) provides details of the algorithm implementation, computational complexity, and stability analysis. Section (\ref{sec:experiment}) compares subgoal planning performance and computation time to an RRT* benchmark planner. Section (\ref{sec:discussion}) provides a discussion of results and future research directions. Finally, section (\ref{sec:conclusion}) provides concluding remarks.
\section{Related Work}
\label{sec:related_work}
This section summarizes recent work investigating characteristics of human cognition in decision making and planning, as well as approaches to computational autonomous vehicle guidance.
The objective is to evaluate ways that current approaches to computational guidance reduce complexity by incorporating characteristics of human behavior. 
Marr's levels of computational analysis \cite{marr1976understanding} provides a framework for comparing algorithms to human cognitive functions in terms of qualities such as the underlying computational theory or objective, task representation, and implementation mechanism. 
The concept of satisficing \cite{simon1972theories} is referenced throughout 
to describe key methods that humans use to achieve sufficiently good sub-optimal solutions despite having limited computational resources.
\subsection{Computational Approaches}
Optimal control theory provides a framework for defining the motion guidance problem, and provides benchmark solutions for comparison of human and other sub-optimal planning approaches.
Optimal control solutions are typically computed analytically \cite{hespanha2009linear, athans2006optimal}, resulting in either a policy function, or complete solution trajectory. 
While the computational objective may be the same in both analytical optimal control and human control, humans learn behavior empirically over many trials. For example, Mettler and Kong \cite{mettler2013mapping} show that guidance policies learned by humans for a motion guidance task take the form of a spatial policy map.

Real-world guidance problems faced by humans typically involve constraints on allowable solution trajectories and nonlinear system dynamics. 
Closed-form solutions to the Hamiltonian function are generally not possible for these problem types. \cite{borrelli2009constrained}. The Hamiltonian formulation however may be used to understand properties of optimal solutions. For general system dynamics and constraint boundaries, a dynamic programming approach must be used \cite{borrelli2005dynamic}. Dynamic programming algorithms are typically implemented by iterating through a large number of finely discretized system states. Task representation as a finely discretized grid is likely incompatible with a biological system because humans and animals do not perceive global quantities such as position or velocity to this degree of precision. In addition, dynamic programming implemented as an iterative process is incompatible with the known neuron-network structure of the brain. Yet, humans easily generate behavior in the presence of nonlinear constraints, suggesting that they utilize an effective sparse representation to achieve their results.

In contrast to analytical approaches, some existing motion planning methods use task representations that are more compatible with biological neural mechanisms. 
In particular, in the field of behavioral robotics \cite{arkin1998behavior}, complex robot behavior results from simpler sensory-motor feedback rules. Artificial potential fields are a method commonly used in robotics path planning, and define a global cost function based on start and goal locations, and obstruction geometry \cite{warren1989global}. A shortcoming of this method is that the potential field can contain local minima, to which solution trajectories may be incorrectly drawn. Forward-chaining \cite{koren1991potential, bell2004forward} is a work-around that avoids local minima by placing intermediate goal points at locations near obstructions. The intermediate goals used in forward chaining are similar to subgoals introduced by Kong and Mettler, but are placed based on an empirical heuristic rather than based on invariants derived from structural characteristics of the agent-environment interactions.

Receding horizon planning \cite{mettler2010agile, primbs2000receding} also uses a concept of an intermediate waypoint to integrate prior global task knowledge with perceptual information from the local environment. This approach allows an agent to incorporate locally-perceived environment constraint information into a global cost-to-go function using a learning process \cite{verma2016computational}.
With this approach, the solution converges to optimal after multiple trials. Local planning in RH guidance, however, can still be computationally expensive. Given the repeated local interactions, such a process should be able to account for previously computed solutions. Verma and Mettler describe a vehicle motion task learning framework based on receding-horizon guidance \cite{verma2016computational}. This framework can be described in terms of Cowan's human information-processing model \cite{cowan2005capacity}, consisting of long-term and working memory elements \cite{baddeley1992working}. Through this framework, Verma and Mettler also describe how the active-waypoint concept in RH guidance relates closely to visual focus of attention in human guidance behavior.

\begin{figure}
    \centering
    \includegraphics[width=.30\textwidth]{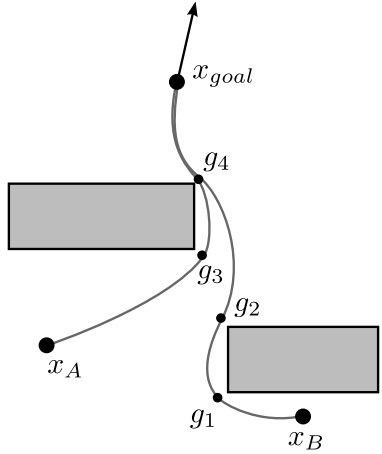}
    \caption{Illustration of subgoal and guidance equivalence classes.}
    \label{fig:equivclasses}
\end{figure}
Spatial cost functions are a plausible task representation for biological systems, but humans routinely navigate in large and complex domains, for which such a representation would require an intractable amount of storage space. To improve the efficiency of policy representations, the following methods utilize the repeating structure present in many control tasks. Tedrake et al. introduced the concept of linear quadratic regulator (LQR) trees \cite{tedrake2010lqr}. LQR trees expand the region of convergence of a regulator by using a hierarchical tree of LQR funnels. Each funnel need only be stable for a defined local region, thereby increasing the global robustness of the system. LQR trees are similar to the subgoal approach presented in this work, which form a global solution by hierarchically combining local solutions. Van Dijk generalizes the identification of subgoals in Markov decision process (MDP) problems \cite{dijk2011grounding} using the concept of an information bottleneck \cite{tishby2000information}. An information bottleneck is a system state that causally separates points along a trajectory; future states are independent of states prior to the bottleneck. This subgoal approach is based on the concept of options \cite{sutton1999mdps}, which are MDP policies that act over a series of states until a termination condition is reached. Options can increase learning rate and decrease policy representation size by re-using policy segments in different parts of a task.

Frazzoli et al \cite{frazzoli1999hybrid, frazzoli2002real} introduced the maneuver automaton (MA), which also takes advantage of task structure in the form of a finite discretization of task dynamics. 
An MA model consists of a library of motion primitive elements describing behavior types, such as straight-line motion or a constant-radius turns. 
In \cite{gavrilets2004human}, the proposed MA contains behavior types inspired by observed human pilot behavior during acrobatic flight. 
Each primitive element may be optimal with respect to vehicle dynamics for a particular guidance objective. 
A sequence of motion primitive elements describe a solution trajectory satisfying environmental constraints.
The MA model simplifies dynamic trajectory planning by limiting actions to the discrete set of motion primitives.
%
In contrast to planners that use task structure to discretize actions, graph-based approaches coarsely discretize the state space a-priori from environment topology. A graph is constructed as a set of nodes and edges within the system state space, and an optimal solution is found using a search algorithm such as A* or D* \cite{russell2003artificial, likhachev2005anytime}. Road-map methods such as Voronoi \cite{bhattacharya2008roadmap} and visibility graphs \cite{lozano-perez1979algorithm} construct graphs by connecting the free space of an environment using heuristics such as safety or path length. Homotopy classes are another graph construction approach that take into account the topological structure of the environment by identifying discrete trajectory equivalence classes describing each possible route around obstacles \cite{bhattacharya2012topological}. Road-map and homotopy class methods offer a plausible model for how humans represent a task environment in terms of the discrete set of free paths or affordances available to the agent \cite{gibson1977theory}. Configuration state graph representations such as these do not however explicitly consider interactions between the environment, system dynamics, and agent perception. To make solutions dynamically feasible, they may be smoothed in post-processing based on vehicle dynamic capabilities \cite{bottasso2008path}. The smoothing process reduces optimality however if it does not consider the global cost function. 
Graph search can be used to compute solutions in the higher-dimension dynamic state space, but becomes computationally expensive due to the high branching factor \cite{bellman1961curse}.

Alternatively, random sampling approaches such as RRT \cite{lavalle1998rapidly} and RRT* \cite{karaman2010optimal} can quickly explore higher-dimension task state spaces to account for dynamic interactions. RRT*-SMART \cite{nasir2013rrt} includes a node-elimination step, resulting in paths that are defined by a minimal set of beacon points. Results show that beacon points converge toward obstacle corners. To take advantage of this observation, CBiRRT* \cite{berenson2009manipulation} uses an optimization step to project sampled nodes onto constraint boundaries, even when the boundaries are not analytically defined. CBiRRT* generates impressive results in high-dimension tasks such as robotic arm motion planning in the presence of torque and configuration state constraints.

Each of these computational approaches are inspired by physical limitations present for both machines and humans. For example, human working memory \cite{cowan2005capacity, baddeley1992working} limits the number of graph nodes that a human can consider simultaneously. Hirtle et al. present evidence that humans may use hierarchical task representations to approximate an environment configuration \cite{hirtle1985evidence}, reducing the number of nodes needed to plan a route. Maneuver automaton, roadmap, and random sampling approaches each mimic aspects of human motion planning that allows them to overcome certain computational limitation. Motion primitive approaches model the way humans plan dynamic trajectories using a library of learned interactions, reducing the set of actions the agent must consider at each step. 
Roadmap methods model feasible paths through the environment independent of a specific task to also limit search dimensionality. Random sampling approaches model a method to trade-off exploration vs planning time, and how humans brains may model the stochastic aspects of a planning task.

\subsection{Human Guidance Modeling}



\subsubsection{Equivalence Classes}
The subgoal planning algorithm presented in this work is based on central concepts in human behavior motion modeling that have been developed in prior work. First, Kong and Mettler recognized patterns present in optimal guidance trajectories within a constrained environment \cite{kong2009general} consisting of subgoals and partitions, and formulate the hypothesis that these patterns emerge in human motion behavior as the result of a learning process \cite{kong2013modeling, mettler2017emergent}. The authors then formalize these subgoals and the associated partition structure. The subgoal concept is related to a method outlined by Shalizi and Crutchfield for determining the hierarchical, causal structure of a process based on equivalence classes within a set of observed signal sequences \cite{shalizi2001computational}. Using this approach, two trajectories belong to the same subgoal equivalence class ($\sim_S$) if they begin at different locations, join together at some point, and remain together until they reach the final goal. For example in Fig. \ref{fig:equivclasses}, trajectories beginning at $\mathbf{x}_A$ and at $\mathbf{x}_B$ belong to the same subgoal equivalence class through subgoal $g_4$. Each partition contains a set of system states for which an optimal trajectory to the goal passes through a common subgoal location. 

\subsubsection{Guidance Model}
The second concept is based on the observation that, upon identifying subgoals, the intermediate trajectory segments share similarities across a problem domain.
The guidance equivalence relation $\sim_G$ relates trajectory segments across partitions, and is based on symmetries in system dynamics. 
Two trajectories belong to the same guidance equivalence class if one can be transformed into the other by applying symmetry transformations. 
For example in Fig. \ref{fig:equivclasses}, trajectory segments $\lbrace g_1, g_2 \rbrace \sim_G \lbrace g_4, x_{goal} \rbrace$ are equivalent through a translation and a rotation, and therefore describe instances of the same subtask solution. 

The guidance equivalence suggests that solution trajectories consist of repeated application of common guidance elements. Furthermore, this symmetry transformation can be applied to segments throughout a trajectory, and across multiple trajectories in a task environment, to aggregate them into a common goal frame. Mettler and Kong show that this aggregate set can then be described by a local spatial motion policy \cite{mettler2013mapping}, in the form of a cost-to-go (CTG) and velocity vector field (VVF) function \cite{kong2011investigation}. The authors validate this model in first-person human guidance behavior \cite{feit2016extraction}.

Spatial cost and policy functions specify the reference velocity as a function of configuration state relative to a subgoal. 
The spatial policy collapses the full-state value function into the configuration space, so as to reduce planning complexity. 
This reduction in dimension is possible because the tracking function regulates higher-order system states to the reference velocity. 
%
\begin{figure}[h]
\centering
\includegraphics[width=0.5\linewidth]{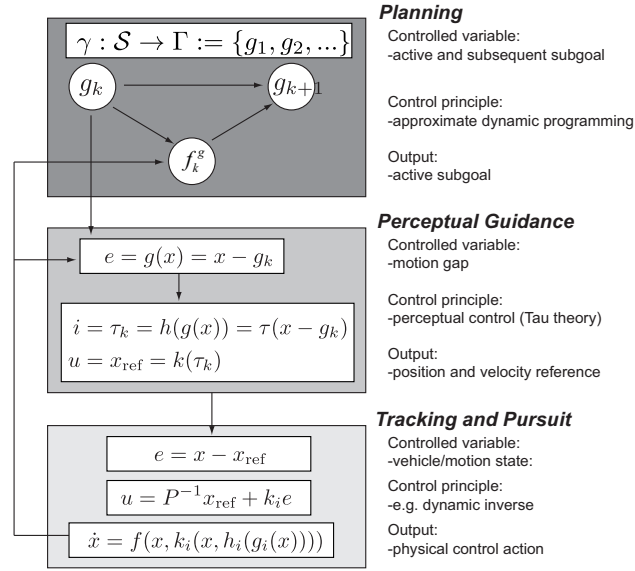}
\caption{Hierarchical model of human motion guidance presented in \cite{mettler2013hierarchical}.}
\label{fig:hierarchical_model}
\end{figure}
\subsubsection{Hierarchical Behavior Model}
The above computational approaches such as RH or RRT planning each discretize the environment in different ways to connect task-level decision making with continuous motion control.
To formalize this interaction between planning levels in human behavior, Mettler et al introduce a hierarchical model connecting interaction patterns \cite{mettler2013hierarchical, mettler2015systems}.
Building on the hierarchical task model described previously, the behavior model proposes three levels for human motion guidance: planning, guidance, and tracking, as illustrated in Fig. \ref{fig:hierarchical_model}. 
Planning involves determining a series of subgoal states $g_k \in \mathcal{F}$, that divide the problem into subtasks. Subgoals are chosen based on task properties such as constraint geometry and system dynamics. Subgoals de-couple discrete action planning from the continuous motion guidance task. The guidance level generates motion elements that connect pairs of subgoals. During guidance, the agent determines a reference velocity trajectory as a function of perceptual information, $\mathbf{v}_{ref} = k(\mathbf{i})$. The agent extracts perceptual information from the environment as a function of system configuration state relative to a subgoal, $\mathbf{i} = h(\mathbf{x} - g_k)$.
Tracking consists of modulating agent or vehicle control inputs to follow the reference velocity trajectory, while accounting for vehicle dynamics. 
Human tracking control is modeled by a combination of system inverse dynamics (i.e. $\textbf{u} = P^{-1}\mathbf{v}_{ref}$) and feedback control ($\mathbf{u} = k(\mathbf{v} - \mathbf{v}_{ref})$). Dynamic inversion such as this is a technique that has been used successfully to mitigate the complexity involved with automatic rotorcraft control \cite{enns2006dynamic}.

\section{Optimal Control Formulation}
\label{sec:theoretical_approach}

The previous section summarized prior work by Kong and Mettler that introduced a model describing the observed structure in human guidance behavior \cite{kong2009general, kong2013modeling}.
This section defines subgoal and related solution structure, building on properties of the constrained optimal control problem formulation. The control problem formulation establishes basic properties of optimal trajectories that are subsequently used to formalize the subgoal model using constrained optimal control theory. These properties are specifically used to formulate conditions that identify candidate subgoal locations based on constraint geometry. These conditions lead to heuristics for the proposed subgoal guidance algorithm presented in subsequent sections.

\subsection{Constrained Optimal Control Formulation}
To identify characteristics of optimal subgoal locations we start 
with the typical Hamiltonian function constrained optimal control formulation:
\begin{equation}
\label{eqn:lagrange}
\mathcal{H} = \lambda(t)^T f(\mathbf{x},\mathbf{u}) - j(\mathbf{x},\mathbf{u}) + \phi c(\mathbf{x})
\end{equation}
In Eqn. (\ref{eqn:lagrange}), the first term accounts for satisfaction of system dynamic constraints, where $\lambda$ is the costate vector, acting as the Lagrange multipliers for system dynamic constraints. The second term accounts for minimization of trajectory cost, and the third term accounts for satisfaction of spatial constraints, i.e. obstacle avoidance. The term $\phi(t) \ge 0$ is the spatial constraint Lagrange multiplier, and indicates the constraint activation at time $t$ along the trajectory. 
Based on Pontryagin's minimization principle, solutions must satisfy the following necessary conditions:
\begin{eqnarray}
\frac{\partial \mathcal{H}}{\partial u_j} = 0
\label{eqn:necessary_condition_1}
& &
\frac{\partial \mathcal{H}}{\partial x_i(t)} = -\dot{\lambda}_i(t) 
\label{eqn:necessary_condition_2}
\end{eqnarray}
The first condition specifies that the optimal solution $u^*(t)$ must minimize the Hamiltonian, and can generally be used to determine $u^*(t)$ as a function of the costate vector $\lambda^*(t)$.
The second condition relates the state trajectory to the costate vector. Similar to the vector optimization case, these condition specify a perpendicularity constraint between the objective function and constraint boundaries, requiring that trajectory points with active spatial constraints occur only when the optimal velocity is tangent to an obstacle boundary.

%

\subsection{Solution Elements and Topology}
To illustrate these conditions, Fig. \ref{fig:constrained_guidance_diagram} depicts an optimization task consisting of goal state $x_g$ and obstacle $\mathcal{O}_i$ with obstacle boundary $\partial \mathcal{O}_i$. The figure contains four example trajectories, $\bar{s}_1$, $\bar{s}_2$, $\bar{s}_{3a}$ and $\bar{s}_{3b}$ generated by optimal policy $\pi_{\mathbf{x}_g}$. Trajectory $\bar{s}_1$ reaches the goal unobstructed using the optimal policy $u = \pi(x)$.
Trajectory $\bar{s}_2$ is also unobstructed, however it coincides with the obstacle boundary between points $a$ and $b$. Trajectory $\bar{s}_{3a}$ is obstructed, as shown by the dashed line passing through $\mathcal{O}_i$. however an alternate path, $\bar{s}_{3b}$ is shown, passing through points $a$ and $b$, which avoids the obstacle.

\subsection{Solution Classification}

To define conditions for subgoal points, we first formally classify the solution types illustrated in Fig. \ref{fig:constrained_guidance_diagram}.
Following the above constrained optimal control formulation and using the Hamiltonian equation, solutions $\mathbf{x}(t)$ for $t \in \left[ 0,T \right]$ take one of three forms: 
\begin{enumerate}
\item Free: $\forall t \in  \left[0,T \right], c(\mathbf{x}(t)) > 0$
\item Constrained: $\forall t \in \left[ 0,T \right], c(\mathbf{x}(t)) = 0$
\item Mixed: $\exists t \in \left[ 0,T \right], g(\mathbf{x}(t)) > 0 \wedge \exists t \in \left[ 0,T \right], g(\mathbf{x}(t)) = 0$
\end{enumerate}
As illustrated in Fig. \ref{fig:constrained_guidance_diagram},
trajectory $\bar{s}_1$ is free, because for all points $\mathbf{x} \in \bar{s}_1$, $c(\mathbf{x}) > 0$. Trajectory $\bar{s}_2$ is mixed because it consist of free and constrained segments. For example, trajectory $\bar{s}_2$ consists of three segments: $\bar{s}(\mathbf{x}_2,\mathbf{a}) \cup \bar{s}(\mathbf{a},\mathbf{b}) \cup \bar{s}(\mathbf{b},\mathbf{x}_g)$. Segments $\bar{s}(\mathbf{x}_2,a)$ and $\bar{s}(\mathbf{b},\mathbf{x}_g)$ are free, because all included points satisfy the constraint. The segment  $\bar{s}(\mathbf{a},\mathbf{b})$ is constrained, representing a special case for which the guidance policy solution trajectory coincides exactly with the constraint boundary: $\mathbf{x} \in \bar{s}(\mathbf{a},\mathbf{b})$, $c(\mathbf{x}) = 0$. Mixed solutions are divided into segments by transition points, such as points $\mathbf{a}$ and $\mathbf{b}$ in Fig. \ref{fig:constrained_guidance_diagram}.

Classifying a solution trajectory in advance as "free", "constrained", or "mixed" simplifies the solution trajectoty. A "free" solution (c.f. trajectory $ \bar{s}(\mathbf{x}_1,\mathbf{x}_g)$ in Fig. \ref{fig:constrained_guidance_diagram}) is specified by the unconstrained guidance policy (Eqn. \ref{eqn:vvf}). 
If a free solution does not satisfy constraints (c.f. trajectory $ \bar{s}_{3a}$ in Fig. \ref{fig:constrained_guidance_diagram}), then the solution must be constrained or mixed. If a solution is mixed, and the transition points between segments can be determined, then the solution trajectory can be assembled from solutions to the free and constrained segments. The next section defines properties of the optimal solution that can be used to determine transition point locations.

\subsection{Equivalence Classes}

Subgoals are first defined generally as a state along a trajectory, $ g_{i} \in \bar{s}(\mathbf{x}_0,\mathbf{x}_g) $. A subgoal $ g_{1} $ along trajectory $ \bar{s}(\mathbf{x}_0,\mathbf{x}_g) $ divides the trajectory into two segments $ \bar{s}(\mathbf{x}_0,g_{1}) $ and $ \bar{s}(g_{1},\mathbf{x}_g) $ as shown in Fig. \ref{fig:subgoal_example}.
Based on the guidance equivalence relation, all states in a segment $\mathbf{x} \in \bar{s}(\mathbf{x}_0,g_{1})$ may be transformed by the group operation $\Psi : \mathbb{R}^4 \rightarrow \mathbb{R}^4$, applying a rotation and translation to the segment such that goal states coincide, $\Psi(g_1) = \mathbf{x}_g$. The resulting transformed segment $\Psi(\bar{s}(\mathbf{x}_0,g_{1}))$ is an instance of an equivalent guidance task to the segment $\bar{s}(g_{1},\mathbf{x}_g)$, in that they are constrained by the same dynamics, and are reaching the same goal state. The guidance equivalence states that, through this transformation, every trajectory solution segment is described by a common spatial guidance policy, $\pi(\mathbf{p}_0)$, where $\mathbf{p}_0$ is the position of the vehicle relative to the goal state.


\subsection{Composite Trajectory}

A composite trajectory is described by a series of subgoals and a guidance policy. The guidance policy mapping defines a free trajectory between a pair of subgoals: $\bar{s}_i^{\pi}(g_{(i-1)}, g_{i}) = \Pi_{g_{i}}(g_{i-1})$. A series of subgoals, $\Gamma(g_0,g_n) = \left\lbrace g_{0}, g_{1}, \dots, g_{n} \right\rbrace $, form a composite trajectory as the union of trajectory segments between subgoals, where $ g_0 = \mathbf{x}_0 $ is the start and $ g_n = \mathbf{x}_g $ is the ultimate goal:
\begin{equation}
\bar{s}(\mathbf{x}_0,\mathbf{x}_g) = \bigcup_{i \in \left[ 1,n \right]} \bar{s}_i^{\pi}(g_{(i-1)}, g_{i})
\label{eqn:composite_trajectory}
\end{equation}
Eqn. \ref{eqn:composite_trajectory} describes a single smooth trajectory connecting $\mathbf{x}_0$ and $\mathbf{x}_g$ because each subgoal is part of both the subsequent and prior trajectory segments.
When a series of two or more trajectory segments are joined, the total cost of the composite trajectory is the sum of the individual segment costs, $ J(\bar{s}(\mathbf{x}_0,\mathbf{x}_g)) = \sum_{i=1}^{n} J(\bar{s}_i(g_{i-1},g_{i})) = J(\Gamma(\mathbf{x}_0,\mathbf{x}_g))$. Such a composite trajectory is piecewise-optimal if it is composed of optimal trajectory segments, $\bar{s}_i(g_{i-1},g_i) = \bar{s}^*(g_{i+1},g_i)$. Based on this, the following Lemma holds:
\begin{nclem}
The total cost of a piecewise-optimal composite trajectory is greater than or equal to the cost of a single optimal trajectory between the start and end points: 
\begin{equation}
\sum_{i=1}^{n} J(\bar{s}^*_i) \ge J(\bar{s}^*(g_0,g_n))
\label{eqn:composite_cost}
\end{equation}
\label{lem:composite}
\end{nclem}
The equality condition in Eqn. \ref{eqn:composite_trajectory} opccurs when all subgoals lie on the optimal trajectory, $g_i \in G \subset \bar{s}^*(g_0,g_n)$.
\begin{figure}
  \centering
  \subfigure[A trajectory divided by a subgoal.]
  {\includegraphics[height=5.0cm]{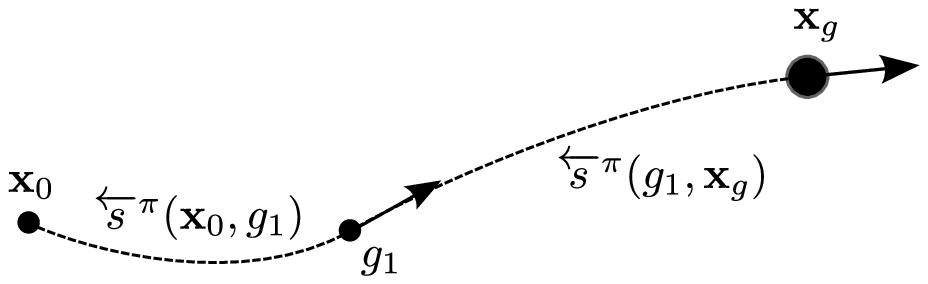}
  \label{fig:subgoal_example}}
  \hfil
  \subfigure[Obstruction avoidance]{\includegraphics[height=5.0cm]{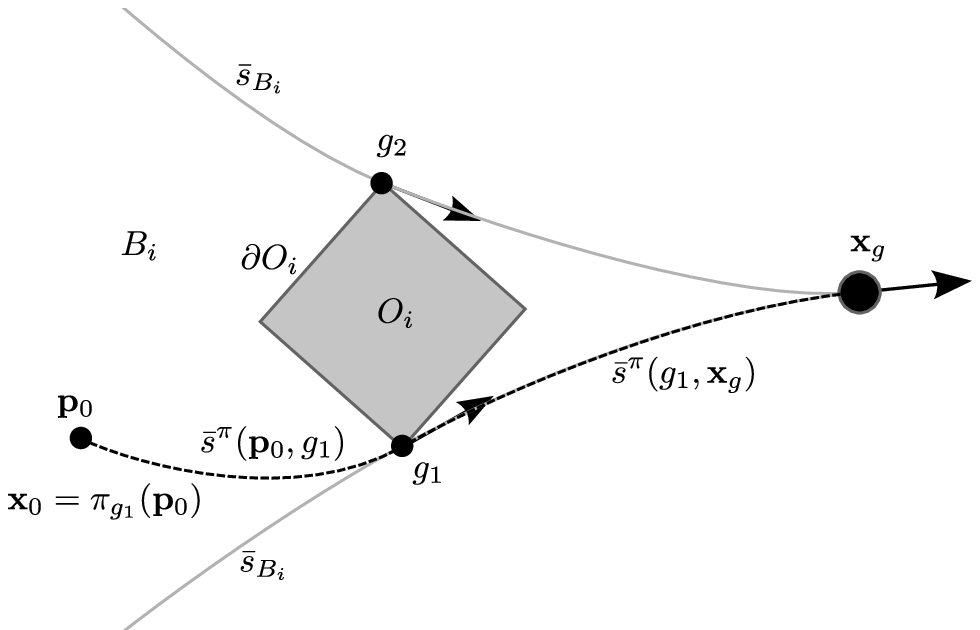}
  \label{fig:solution_traj}}
 \caption{Solution Trajectory}
 \label{fig:solution_trajectories_0}
\end{figure}

\begin{figure}[h]
\centering
\includegraphics[width=0.45\linewidth]{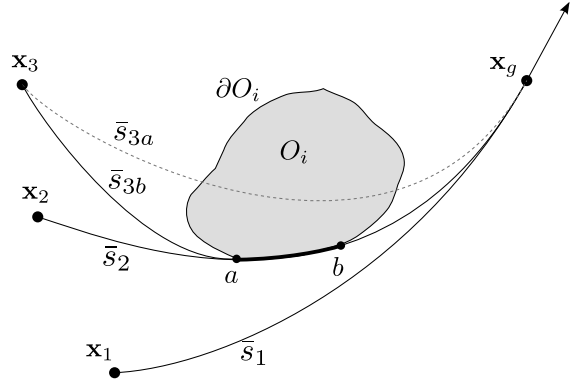}
\caption{The constrained optimal control problem setup, illustrating obstacle free, constrained, and mixed solutions}
\label{fig:constrained_guidance_diagram}
\end{figure}

\section{Structural Properties and Heuristics}
\label{sec:structural_properties}

The previous section showed how constraint transition points simplify the constrained optimal control problem. This section investigates transition point properties by defining bounding trajectories, which are evident in constrained optimal control solutions. These transition point properties are then used to infer necessary conditions for feasible motion planning subgoal candidates. Finally, the resulting subgoal properties are used to outline a subgoal planning procedure.

\subsection{Constraint Structure}
\label{sec:constraint_structure}

A constrained region $B_i(\mathbf{x}_g,\pi,O_i)$ contains all points $\mathbf{p}$ such that a trajectory from $\mathbf{p}$ to $\mathbf{x}_g$ using policy $\pi$, $\bar{s}^{\pi}(\mathbf{p},\mathbf{x}_g)$, does not satisfy $O_i$, as illustrated in Fig. \ref{fig:solution_traj}. The boundary of $B_i(x_g,\pi,O_i)$ is $S_{B_i}(\mathbf{x}_g,\pi,O_i)$, and consists of points $\mathbf{p}$ such that the trajectory $\bar{s}^{\pi}(\mathbf{p}_0,\mathbf{x}_g)$ intersects only the constraint boundary $\partial O_i$, i.e.:
\begin{itemize}
\item $\exists \mathbf{p} \in \bar{s}^{\pi}(\mathbf{p}_0,\mathbf{x}_g) \mid  c_i(\mathbf{p}) = 0$
\item $\forall \mathbf{p} \in \bar{s}^{\pi}(\mathbf{p}_0,\mathbf{x}_g) \mid  c_i(\mathbf{p}) \ge 0$
\end{itemize}
A trajectory $\bar{s}^{\pi}(\mathbf{p}_0,\mathbf{x}_g)$ is a bounding trajectory $\bar{s}_{B_i}(\mathbf{x}_g,\pi,O_i)$ if $\mathbf{p}_0 \in S_{B_i}(\mathbf{x}_g,\pi,O_i)$, i.e. if $\mathbf{p}_0$ is on the boundary of the constrained region. Conversely, all trajectories starting outside of the constrained region, $\bar{s}^{\pi}(\mathbf{p}_0,\mathbf{x}_g)$ for $\mathbf{p}_0 \notin B_i $, satisfy $O_i$.

\subsection{Necessary Conditions}

\begin{figure}[h]
\centering
\subfigure[Infeasible subgoal domains $g'_1 \in \mathcal{D}_1$, $g'_2 \in \mathcal{D}_2$, and $g'_3 \in \mathcal{D}_3$.]{
\includegraphics[height=5.0cm]{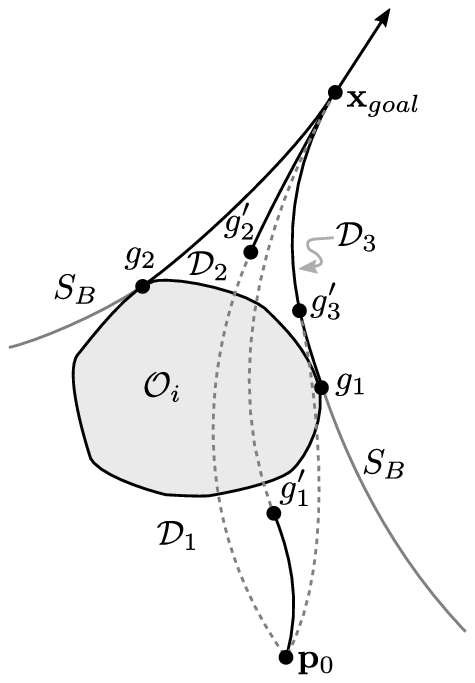}
\label{fig:necessary_conditions_1}}
\hfil
\subfigure[Suboptimal subgoal domains $g'_4 \in \mathcal{D}_4$ and $g'_5 \in \mathcal{D}_5$.]{
\includegraphics[height=5.0cm]{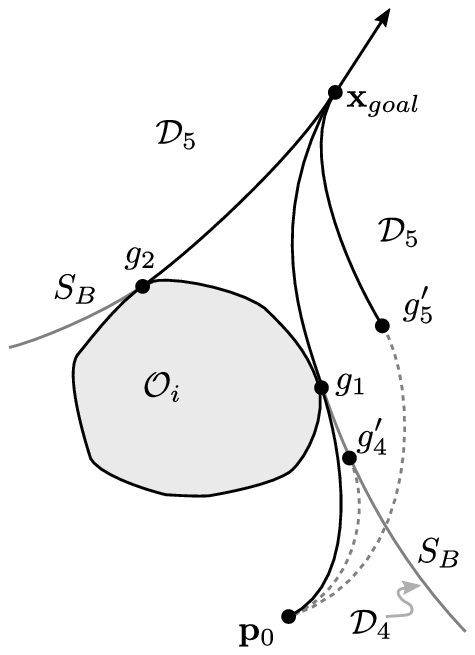}
\label{fig:necessary_conditions_2}}
\hfil
\subfigure[Continuous constraint boundary. Subgoal where $\pi(\mathbf{x})$ is perpendicular to $\nabla c(\mathbf{x})$]{
\includegraphics[height=5.0cm]{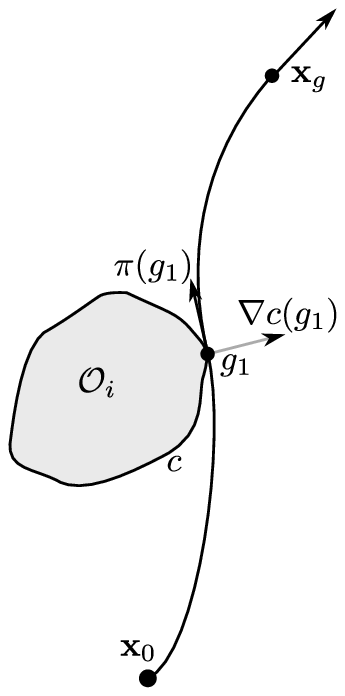}
\label{fig:continuous_boundary}}
\hfil
\subfigure[Piecewise-smooth constraint boundary. Subgoal where $\pi(\mathbf{x})$ is in range $ \left( \nabla c_1(\mathbf{x}), \nabla c_2(\mathbf{x}) \right) $]{
\includegraphics[height=5.0cm]{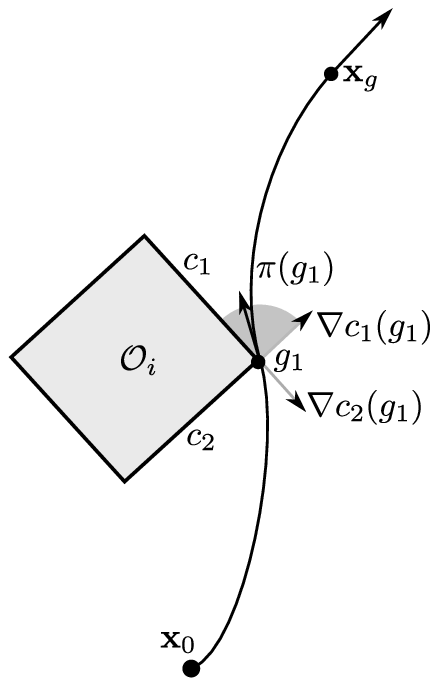}
\label{fig:piecewise_boundary}}
\caption{Example subgoal placement based on constraint geometry.}
\label{fig:subgoal_placement}
\vspace{-5mm}
\end{figure}


The constrained optimal control problem is now reformulated in terms of subgoals and bounding trajectories. The goal is to determine the sequence of subgoals, $\Gamma(g_0, g_n) = \{g_0, g_1, \dots, g_n\}$ that specify the piece-wise optimal minimum-cost composite trajectory that satisfies constraints. We first consider the case of a single obstruction, $O_1$. The constrained optimal solution trajectory in this case contains a single subgoal, $\Gamma(\mathbf{x}_0, \mathbf{x}_g) = \{\mathbf{x}_0, g_1, \mathbf{x}_g\}$, as illustrated in Fig. \ref{fig:solution_traj}. The objective is to choose the optimal subgoal state $g^* \in \mathcal{X}$ such that total trajectory cost is minimized, and each segment satisfies constraints:
\begin{eqnarray}
\label{eqn:optimal_subgoal}
g_1 = \argmin_{g \in \mathcal{X}} J^*_{g}(\mathbf{x}_0)) + J^*_{\mathbf{x}_g}(g)) \\
\text{such that }\bar{s}(\mathbf{x}_0,g) \notin \mathcal{O}_E \wedge \bar{s}(g,\mathbf{x}_g) \notin \mathcal{O}_E \nonumber
\end{eqnarray}
Eqn. \ref{eqn:optimal_subgoal} is a minimization over a continuous state-space domain, $\mathcal{X}$. To solve this efficiently, properties of the constrained optimal control problem given in Sec. \ref{sec:theoretical_approach} are used to state necessary conditions that reduce the subgoal candidate domain, $\mathcal{G} \subset \mathcal{X}$, enabling efficient determination of the optimal subgoal.

\begin{ncdef}
The subgoal candidate set $\mathcal{G}$ is the intersection of bounding trajectories and constraint boundaries:
\begin{equation}
\mathcal{G}(\mathbf{x}_g, \pi, O_i) := S_B(\mathbf{x}_g, \pi, O_i) \cap \partial O_i
\label{eqn:admissible_subgoals_2}
\end{equation}
\end{ncdef}

\begin{nclem}
Given the constrained optimization problem outlined above $\langle \mathbf{x}_g, \pi, O_i \rangle$, the optimal subgoal point $g^*$ defining the minimum-cost composite solution trajectory $\bar{s}^{\pi}(\mathbf{p} \in B_i, g^*, \mathbf{x}_g)$ is contained within the set of subgoal candidates for obstacle $O_i$,  i.e. $g^* \in \mathcal{G}(\mathbf{x}_g, \pi, O_i)$.
\label{lem:optimal_subgoal_candidate}
\end{nclem}

\begin{proof}
Lemma \ref{lem:optimal_subgoal_candidate} is proven by showing that any other subgoal locations $ g' \notin \mathcal{G}$ are either infeasible or have a higher cost than at least one subgoal candidate in $\mathcal{G}$. To show this, the spatial domain is divided into regions $\mathcal{D}_1,\dots, \mathcal{D}_5$ surrounding each subgoal candidate, $g_i \in \mathcal{G}$, as shown in Figs. \ref{fig:necessary_conditions_1} and \ref{fig:necessary_conditions_2}. First, Fig. \ref{fig:necessary_conditions_1} illustrates how subgoal locations within the bounded region, both behind the obstacle ($\mathcal{D}_1$), in front of the obstacle ($\mathcal{D}_2$), or on the boundary in front of the subgoal candidate ($\mathcal{D}_3$) are infeasible because there is either no free optimal path from the subgoal to the goal, or no feasible path from the start to the subgoal.

The remaining feasible subgoal locations occur either along the bounding trajectory behind the subgoal candidate ($\mathcal{D}_4$), or outside of the bounding trajectory ($\mathcal{D}_5$). If the subgoal is placed along the bounding trajectory behind the subgoal candidate ($\mathcal{D}_4$), then this solution is equivalent to a trajectory with two subgoals: $\bar{s}^{\pi}(\mathbf{p}_0, g'_4, g_1, \mathbf{x}_{goal})$ as shown in Fig. (\ref{fig:necessary_conditions_2}). In this case, based on Lemma \ref{lem:composite}, the segment from $\mathbf{p}$ to $g_1$ can be replaced with a single guidance segment with equal or lesser cost by eliminating subgoal $g'_4$.

Finally, any subgoal in $\mathcal{D}_5$ specifies a solution that intersects the bounding trajectory, as illustrated in Fig. \ref{fig:necessary_conditions_2}. At this intersection point, two paths to the goal are available. Since any sub-segment of an optimal trajectory must also be optimal, an intersection point presents a contradiction of the optimality principle. Any subgoal that defines a path that intersects the optimal bounding trajectory therefore must have a higher cost from the intersection point to the goal, and cannot be optimal.
\end{proof}
Note that based on the principle of optimality, subgoal candidates $g_i \in \mathcal{G}$ have the property of defining solution trajectories that are both feasible and do not intersect any other feasible, optimal trajectories to the goal. Lemma 1 leads to the following conditions defining feasible subgoal candidates:
\begin{nccon}
Subgoal candidates lie on constraint boundaries: $\mathcal{G} \subset \partial O_i$.
\label{cond:constraint_boundary}
\end{nccon}
\begin{nccon}
The velocity specified by the VVF associated with $\pi$ at a subgoal candidate must be tangent to the obstacle boundary.
\begin{equation}
\label{eqn:nc_2}
\mathcal{G} = \{ g | c_i(g_i \pm \epsilon \cdot \pi(g_i)) > 0 \}
\end{equation}
\label{cond:tangent_velocity}
\end{nccon}
Conditions \ref{cond:constraint_boundary} and \ref{cond:tangent_velocity} are illustrated in Fig. \ref{fig:continuous_boundary}. If a subgoal lies at a convex corner where two constraint boundary segments, $c_{1}$ and $c_{2}$ meet, then a subgoal is admissible over a range of velocity directions as illustrated in Fig. \ref{fig:piecewise_boundary}:
\begin{eqnarray}
\label{eqn:velocity_vector_condition_corner}
\nabla c_{i1} \cdot \pi(g_i) \ge 0 \vee \nabla c_{i1} \cdot \pi(g_i) \ge 0\\
\nabla c_{i2} \cdot -\pi(g_i) \ge 0 \vee \nabla c_{i2} \cdot -\pi(g_i) \ge 0 \nonumber
\end{eqnarray}

\subsection{Solution Subgoal Example}
\label{sec:subgoals}
The constrained optimal control problem with a single obstruction is illustrated in Fig. \ref{fig:solution_traj}. First, note that the direct trajectory between $\mathbf{x}_0$ and $\mathbf{x}_g$ using the optimal policy, $ \bar{s}(\mathbf{x}_0,\mathbf{x}_g) = \Pi_{\mathbf{x}_g}(\mathbf{x}_0)$ does not satisfy constraints: $\exists \mathbf{x} \in \bar{s}(\mathbf{x}_0,\mathbf{x}_g) \mid \mathbf{x} \in O_i$. Because the direct trajectory is infeasible, an alternate composite solution trajectory must be determined consisting of free or constrained segments. To find these segments, the conditions in Eqn. \ref{eqn:admissible_subgoals_2} are used to find a discrete set of admissible subgoals transition points, $\mathcal{G} = \{g_1, g_2\}$. In this case, it is assumed that the optimal trajectory consists of two segments, joined at a single subgoal. The optimal subgoal is chosen from $\mathcal{G}$ that minimizes the total trajectory cost based on Eqn. \ref{eqn:optimal_subgoal}. In Fig. \ref{fig:solution_traj}, $g_1$ is depicted as the optimal subgoal, and the optimal solution trajectory is illustrated as $\bar{s}^*_1(\mathbf{x}_0, g_1) \cup \bar{s}^*_2(g_1,\mathbf{x}_g)$.

In environments with multiple obstructions, the solution trajectory may require multiple subgoals. The optimization problem becomes that of selecting an optimal sequence of subgoals, $\Gamma_{\pi}^* = \{g^*_1, g^*_2, \dots, g^*_n\}$ that define a minimum cost path to the goal using trajectory segments generated by policy $\pi$. 
\begin{equation}
\Gamma_{\pi}^* = \argmin_{\Gamma_i \subset \mathcal{G}} J_{\pi}(\Gamma_i))
\end{equation}
The principle of optimality \cite{bellman1957dynamic} states that any sub-trajectory $\bar{s}(g_i, g_{i+k})$ for $g_i,g_{i+k} \in \Gamma$ must be an optimal trajectory between endpoint subgoals $g_i$ and $g_{i+k}$. Hence, the optimal subgoal sequence can be defined recursively:
\begin{equation}
\label{eqn:mult_subgoals_recursive}
\Gamma_{\pi}^*(g_0, g_n) = \lbrace \Gamma_{\pi}^*(g_0, g_{n-2}), \argmin_{g_{n-1} \in \mathcal{G}} J_{\pi}^*\left(\Gamma_{\pi}^*(g_0, g_{n-2}), g_{n-1}, g_n\right), g_n \rbrace
\end{equation}
Eqn. \ref{eqn:mult_subgoals_recursive} reduces the problem size by one, by solving a single-subgoal problem as presented in Eqn. \ref{eqn:optimal_subgoal}. In practice, a solution $\Gamma_{\pi}^*$ is found using dynamic programming.

\subsection{Additional Properties}
Given a guidance task with a single obstacle, $\langle \mathbf{x}_g, \pi, O_i \rangle$ and subgoal candidate set $\mathcal{G}$, the constrained region $B_i$ is further divided into partitions, such that all optimal trajectories beginning within a partition $ P(g_{j} \in \mathcal{G}) \subset B_i$ converge to subgoal $g_j$ on the optimal path to the goal $\mathbf{x}_g$. A partition $ P_j(g_{j}) $ is defined by bounding trajectories $ S_{B_i} $, and a separatrix, $ T_{B_i} $, which is a set of points for which the total trajectory cost is equal for two or more different subgoals $ g_i,g_j \in \mathcal{G}$:
\begin{equation}
T_{g_i,g_j} = \lbrace \mathbf{p} \mid J^{\pi}_{g_{i}}(\mathbf{p}) + J^{\pi}_{\mathbf{x}_g}(g_{i}) = J^{\pi}_{g_{j}}(\mathbf{p}) + J^{\pi}_{\mathbf{x}_g}(g_{j}) \rbrace
\end{equation}
Partition boundaries $\partial P = T_B \cup S_B$ are switching surfaces, separating regions in which unique subgoals are optimal. Separatrices $ T_B $ act as repelling manifolds, separating initial states that move toward different subgoals. Bounding trajectories $S_B$ are attracting manifolds, such that initial states on either side of the manifold result in nearly the same optimal trajectory. Importantly, partition sets define regions of local independence between initial position $\mathbf{p}_0 \in P$ and subgoal location $g^* \in \mathcal{X}$, forming an "information bottleneck" \cite{braun2010structure}, and allowing a static set of subgoals to be used across a task environment. Subgoal and partition properties are consistent with equivalence relations introduced by Kong and Mettler \cite{kong2009general}. The set of initial configurations in a partition $\mathbf{p}_0 \in P_i$ belong to the same subgoal equivalence class through subgoal $g_i$. By this equivalence, determining the partition that an initial state belongs to fully specifies the remaining trajectory to the goal. Furthermore, partitions across a task are related through the guidance equivalence relation, because they each represent a similar subtask of reaching a subgoal state.

Subgoals and the associated partitions discretize the task, transforming the continuous trajectory planning problem into a discrete planning problem, similar to road-map motion planning approaches. To use a graph-search algorithm, nodes must satisfy the Markov property \cite{russell2003artificial}. 
This property requires that each next subgoal cannot depend on any previous subgoal along the trajectory, but only on the current state. 
The subgoal properties above satisfy the Markov condition; the velocity vector of a subgoal $g_1$, $v_{g_1} = \pi(p_{g_1}, x_{g_2})$, depends on the next subgoal state through the guidance policy, and the subgoal position depends on the next subgoal state through the necessary conditions (Eqn \ref{eqn:admissible_subgoals_2}), but is independent of the path prior to that point.
Note that the Markov condition is consistent with the subgoal equivalence observed in human guidance behavior, stating that trajectories that meet at a subgoal remain together until they reach the final goal.
This subgoal structure results in partitions that are hierarchically included, each partition being a subset of a partition that is closer to the final goal.

\subsection{Stability}

\subsubsection{Approach}
System stability is important for both humans and computational systems to guarantee that a system converges to a target state \cite{khalil2002nonlinear}, even when a sub-optimal, satisficing solution is used. For the motion guidance task, the system must reach the target in a bounded time interval, $t \in \left[0, T\right]$. Moulay and Perruquetti present Lyupanov-based criteria for finite-time stability \cite{moulay2005lyapunov}:
\begin{equation}
T(\mathbf{x}_0) = \int_{V(\mathbf{x}_0)}^{0} \frac{d \xi}{\dot{V}(\bar{s}_{\mathbf{x}_g}(\mathbf{x}_0,\theta(\xi))} < + \infty
\label{eqn:finite_time_stability}
\end{equation}
In Eqn. \ref{eqn:finite_time_stability}, $T(\mathbf{x}_0)$ is the settling-time, or time-to-go of initial state $\mathbf{x}_0$. The mapping $\theta : V(\bar{s}_{x_g}(\mathbf{x}_0,t)) \rightarrow t$ is inverse cost, relating Lyuponov function value to settling time. In Eqn. \ref{eqn:finite_time_stability} the integrand expresses the differential time-to-go, $dt$, in terms of the Lyupanov function value dummy variable $\xi$. 
Eqn. \ref{eqn:finite_time_stability} can be used to demonstrate finite-time stability as follows. Assume a bounding function $g(\xi)$ exists such that $g(\xi) \in L^1\left(\left[0, \sup_{\mathbf{x} \in \mathcal{V}} V(\mathbf{x}) \right]\right)$. The system is finite-time stable if for all $\mathbf{x} \in \mathcal{V} - \{\mathbf{x}_g\}$, and all $\xi \in \left[0, V(\mathbf{x})\right]$:
\begin{equation}
\frac{-1}{\dot{V}((\bar{s}_{\mathbf{x}_g}(\mathbf{x}_0, \theta(\xi)))} \le g(\xi)
\label{eqn:finite_time_bounded}
\end{equation} 
As described in \cite{moulay2005lyapunov}, Eqn. \ref{eqn:finite_time_bounded} shows that when $\dot{V}(\mathbf{x}) \le -c(V(\mathbf{x}))^{\alpha}$, for  $c > 0$, $\alpha \in \left[0, 1\right]$, the settling time of an initial state $\mathbf{x}_0 \in \mathcal{V}$ is bounded by:
\begin{equation}
T(\mathbf{x}_0) \le \frac{V(\mathbf{x}_0)^{1-\alpha}}{c(1 - \alpha)}
\label{eqn:finite_time_bound}
\end{equation}
\subsubsection{Planning Stability}
The stability of a constrained subgoal planning problem, $\left\langle \mathbf{x}_0, \mathbf{x}_g, \mathcal{O}_E, \Pi \right\rangle$ is considered with respect to the evolution of a solution sequence of subgoals, $\Gamma = \{g_0, g_1, \dots \}$ with $g_0 = \mathbf{x}_0$. A Lyupanov function, $V(g_k) := J^*_{\pi}(\Gamma(g_k,g_n)) = \sum_{i=k}^{n-1} J^*_{\pi}(g_i, g_{i+1})$ is defined as the total cost incurred by the sequence of guidance elements connecting each pair of subgoals, $\langle g_i, g_{i+1} \rangle$ in sequence $\Gamma$. For this discrete-time subgoal transition process, Eqn. \ref{eqn:finite_time_stability} is expressed as:
\begin{equation}
T(g_0) = \sum_{g_k \in \Gamma} \frac{\Delta V(g_k)}{\Delta V(g_k) / \Delta T(g_k)} = \sum_{g_k \in \Gamma} \Delta T(g_k) < + \infty
\label{eqn:planning_stability}
\end{equation}
In Eqn. \ref{eqn:planning_stability}, $\Delta V$ is the change in Lyapunov value over the transition from subgoal $g_{k-1}$ to $g_k$, and $\Delta T$ is the corresponding change in time-to-go, i.e. $\Delta V / \Delta T = (V(g_k) - V(g_{k-1})/(T(g_k) - T(g_{k-1})$. Based on Eqn. \ref{eqn:planning_stability}, finite-time planning stability requires that both $\Delta T(g_k) < + \infty$ and $ \Delta V(g_k) < + \infty$ for all $g_k \in \Gamma$. Practically, these two conditions are met if the guidance policy is finite-time stable for each pair of subgoals in the plan, and if the plan reaches the goal using a finite number of subgoals. Planning stability is of greater concern for tasks in unknown environments, where the agent learns about subgoal connections and estimates subgoal cost as they move. For example, the agent must use a planning strategy that avoids entering a cycle. Guidance policy stability is addressed later in the paper.

\section{Planning Implementation}
\label{sec:implementation}
The previous section introduces elements needed to formulate the constrained optimal control problem as a graph search problem. These element include the necessary conditions specifying a discrete, finite set of optimal subgoal candidates, and the recursive approach to determining an optimal subgoal sequence defined in Eqn. \ref{eqn:mult_subgoals_recursive}. This section describes the subgoal graph planning approach. In this approach, graph nodes consist of admissible subgoal candidates, $\mathcal{G}$, and the guidance policy $\Pi$ provides feasible edges between nodes. 
An optimal graph search algorithm finds a sequence of subgoals, $\Gamma = \{\mathbf{x}_0, g_1, \dots, g_n\}$ that specify a piecewise-optimal solution trajectory.

\subsection{Subgoal Planning}
Edge costs for graph planning are often computed as the spatial distance between nodes, however in dynamic tasks, path cost also depends on velocity and other higher-order states. Incorporating these higher-order states causes node cost to become dependant on prior nodes along a path, breaking the Markov condition that ensures a correct solution.
In the proposed subgoal guidance approach, the guidance policy $\Pi$ maintains the Markov condition by restricting node velocity to a function of the next node state. The guidance policy thereby creates independence between the current node and any prior vehicle state, as long as the vehicle can track the reference velocity with sufficiently small error.
This approach creates a backward dependence between nodes, from the ultimate goal towards the start.

To account for the backwards dependence, backwards A* graph search \cite{ferguson2005guide} is used to determine solution subgoal sequences as described in Fig. \ref{alg:astar}.
Beginning at the goal node, the \textbf{getNeighbors} function in in Fig. \ref{alg:getneighbors} performs backward expansion, returning a list of subgoals from which subgoal $g$ can be reached using a feasible nominal trajectory, along with the cost associated with each trajectory. 
Trajectory feasibility and cost is computed in the \textbf{predictTrajectory} function, which simulates the trajectory $\bar{s}(\mathbf{x}_{start},\mathbf{x}_{goal})$ using guidance policy $\dot{\mathbf{x}} = \pi^*(\mathbf{x})$. 
The function \textbf{isFree($\bar{s}$)} determines the logical $I = \bar{s} \cap \mathcal{O}_E$, indicating whether the predicted trajectory $\bar{s}$ avoids obstacles.

During SGP execution, the majority of computation time is spent in the \textbf{predictTrajectory} function, so reducing the number of trajectory predictions is key to decreasing search time. Two strategies based on satisficing reduce the number of paths that must be simulated. 
First, \textbf{getNeighbors} considers only the $n_{limit}$ neighboring admissible subgoals with least heuristic costs. When $n_{limit} \ge N_{tot}$, all subgoals are considered, and the algorithm is optimal. 
When $n_{limit} < N_{tot}$, computation time is reduced, but optimal subgoals may be missed if the heuristic significantly underestimates actual cost. 
The second strategy is to prune out neighbor subgoal branches that are unlikely to be part of a minimum-cost path based on their incremental increase in cost from previously explored neighbors. 
During the node expansion phase, after \textbf{getNeighbors} expands the $n_{min}$ lowest-cost nodes in the open list, it only explores additional nodes that are within cost tolerance $\epsilon$ of the previously explored neighbor.
When $\epsilon = \infty$, subgoals are never pruned, and optimality is preserved. 


\begin{figure}
\begin{minipage}[t]{0.5\linewidth}
\begin{algorithm}[H]
 \label{alg:astar}
 \KwData{$\mathbf{x}_{goal}$,$\mathbf{x}_{start}$,$G_{all}$}
 \KwResult{$\Gamma_{\mathbf{x}_g}$}
 $G_{open} = $ PriorityQueue()\;
 $goal.value = \mathbf{x}_g$\;
 $goal.key = h(\mathbf{x}_{start},\mathbf{x}_{goal})$\;
 $\text{insert}(G_{open}, goal)$\;
 \While{$\mathbf{x}_{start} \notin G_{open}$}{
  $g_{node} = \text{popMin}(G_{open})$\;
  $G_{neighbors} = \text{getNeighbors}(g_{node})$\;
  \ForEach{$g_{i} \in G_{neighbors}$}{
   \If{$g_{i}.cost < G_{all}[g_i.id].cost$}{
    $G_{all}[g_i.id].cost = g_i.cost$\;
    $G_{all}[g_{i}.id].next = g_{node}$\;
    \uIf{$g_i \in G_{open}$}{
     $\text{updateKey}(G_{open}, g_i, g_i.cost)$\;
    }\Else{$\text{insert}(G_{open}, g_i, g_i.cost)$}
   }
  }
 }
 $g = \mathbf{x}_0$\;
 $\Gamma_{\mathbf{x}_g} = \{g\}$\;
 \While{$g.next \ne \emptyset$}{
 $g = G_{all}[g.next]$\;
 $\Gamma_{x_g} = \Gamma_{\mathbf{x}_g} \cup \{g\}$\;
 }
\end{algorithm}
\end{minipage}
\begin{minipage}[t]{0.5\linewidth}
\begin{algorithm}[H]
 \label{alg:getneighbors}
 \KwData{$g_{node}$, $G_{all}$, $E_{all}$, $\mathcal{O}_E$, $\Pi$}
 \KwResult{$G_{neighbors}$}
 $H = \text{PriorityQueue}()$\;
 $G_{edges} = \text{findEdgeSubgoals}(E_{all}, g_{node})$\;
 \ForEach{$g_{i} \in G_{all} \cup G_{edges}$}{
  $h_i = h(g_{start},g_i) + h(g_i,g_{node}) + g_{node}.cost$\;
  $g.value = g_i$\;
  $g.key = h_i$\;
  $\text{insert}(H, g)$
 }
 $G_{neighbors} = \emptyset$\;
 \While{$|G_{neighbors}|  < n_{limit}$}{
  $\{h^*_i, i^*\} = \text{popMin}(H)$\;
  \If{$|G_{neighbors}| < n_{min}$ or $h^*_i - h_{last} < \epsilon$}{
   $\bar{s}_{test} = \text{predictTrajectory}(g_{i^*},g_{node},\Pi)$\;
   \If{isFree($\bar{s}_{test}$, $\mathcal{O}_E$)}{
    $g_{i^*}.cost = g_{node}.cost + \text{cost}(\bar{s}_{test})$\;
    $G_{neighbors} = G_{neighbors} \cup g_{i^*}$\;
    \If{$g_i \notin G_{all}$}{$G_{all} = G_{all} \cup g_i$}
   }
  }
 }
\end{algorithm}
\end{minipage}
\par
\begin{minipage}[b]{0.5\linewidth}
\caption{backwardsAstar()}
\end{minipage}
\begin{minipage}[b]{0.5\linewidth}
\caption{getNeighbors()}
\end{minipage}
\end{figure}

\begin{figure}
\begin{algorithm}[H]
\label{alg:findedge}
\KwData{$E_{all}$,$g_{node}$}
\KwResult{$G_{edges}$}
$G_{edges} = \emptyset$\;
\ForEach{$E_i \in E_{all}$}{
 $(v_1, v_2) = E_i$.vertices()\;
 $\hat{e} = v_1 - v_2$\;
 $\psi_E = \tan^{-1}(-\hat{e}_1, \hat{e}_2)$\;
 $p(s) = v_1 + s * \hat{e}$\;
 $s^* = \text{solve}(\pi_{g_{node}}(p(s)) == \psi_E 
 \rightarrow s)$\;
 \If{$s^* \in \left[0,1\right]$}{
  $G_{edges}$.append($p(s^*)$)\;
 }
}
\end{algorithm}
\caption{getEdgeSubgoals()}
\end{figure}


Conditions for admissible subgoal candidates are defined in Eqn. \ref{eqn:admissible_subgoals_2}, based on constraint boundaries and velocity vector direction. When \textbf{getNeighbors} is called, it searches for two types of subgoal candidates. First, obstruction vertices are considered. In example cases presented in this paper, obstructions are convex, with piecewise linear boundaries. As shown in Fig. \ref{fig:piecewise_boundary}, obstruction vertices allow for a range of possible velocity vector directions that satisfy the necessary conditions. As a result, admissible subgoals almost always occur at obstruction vertices. Subgoals less frequently occur at points along continuous constraint boundary segments where the velocity vector specified by the nominal policy $v = \pi_{g}(\mathbf{x})$ is parallel to the constraint boundary, as in Fig. \ref{fig:continuous_boundary}. To accommodate this case, \textbf{getNeighbors} also checks each obstruction edge (\textbf{getEdgeSubgoals()}) to determine if a point along it satisfies Eqn. \ref{eqn:admissible_subgoals_2} and should be included as a subgoal candidate.

In this implementation, obstacles boundaries $\partial O$ are defined in the workspace, $p \in \mathcal{W}$, such that $c_i(p) = 0$. Spatial constraints in any real vehicle or robotic system however are a function of system configuration (C-space), i.e. $c_i(q) = 0$, because they depend on, for example, vehicle geometry and orientation, or robot arm joint configuration, in addition to just the end effector position. The assumption used in this implementation is that the C-space constraints can be tightly overbounded by workspace constraints. If $p(q)$ is a projection of a configuration $q = \{p, \psi\}$ into workspace point $p$: $c_i(p(q)) = c_i(q) + \epsilon$, where $0 \le \epsilon \le K$ for all configurations $q \in \mathcal{Q}$ and some constant $K$. In addition to system configurations, the tolerance $\epsilon$ takes into consideration obstacle clearance required due to the physical size of the robot, as well as any tolerance needed to ensure a safe trajectory while accounting for uncertainties in system control performance. As a result, subgoals placed on obstacle boundaries defined as above include the required offset or clearance required to allow a feasible, safe trajectory to pass through. 

\subsection{Vehicle Dynamics and Guidance Policy}
\label{sec:control_policy}

The subgoal guidance implementation uses a unicycle vehicle model with lateral and forward acceleration limits:

\begin{equation}
\left[ \dot{x} \, \dot{y} \, \dot{\psi} \, \dot{v}  \right]^T
=
\left[ v \cos \psi \,\, v \sin \psi \,\, \min(u_{lat}/v, \omega_{max}) \,\, f(v, u_{lon}) \right]^T
\label{eqn:vehicle_dynamics}
\end{equation}

This model was used in prior human guidance behavior investigations \cite{feit2015experimental, feit2016extraction} because it incorporates challenges typical in dynamic, human motion control tasks. For example, because turn rate is limited by forward speed, the agent must plan in advance to make successful turns.

This system has three configuration states, $x$, $y$, and $\psi$. The first two rows in Eqn. \ref{eqn:vehicle_dynamics} define a non-holonomic constraint on vehicle velocity state through vehicle heading, $\dot{y}/\dot{x} = \tan \psi$. In this guidance implementation, vehicle size is assumed to be small with respect to constraint dimensions. This assumption allows vehicle heading to be considered as an action that is freely controlled independent of constraints - i.e. on a constraint boundary, or at a subgoal, vehicle heading can take any value. 
%
\begin{figure}
\vspace{5mm}
 \centering
   \subfigure[Guidance geometry.]{\includegraphics[height=4.5cm]{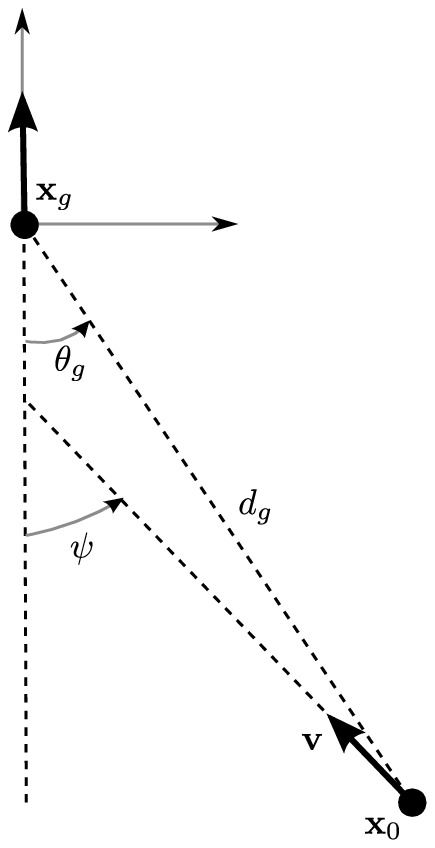}
    \label{fig:guidance_geometry}}
    \hfil
  \subfigure[Observed guidance behavior and perceptual guidance policy.]{\includegraphics[height=4.5cm]{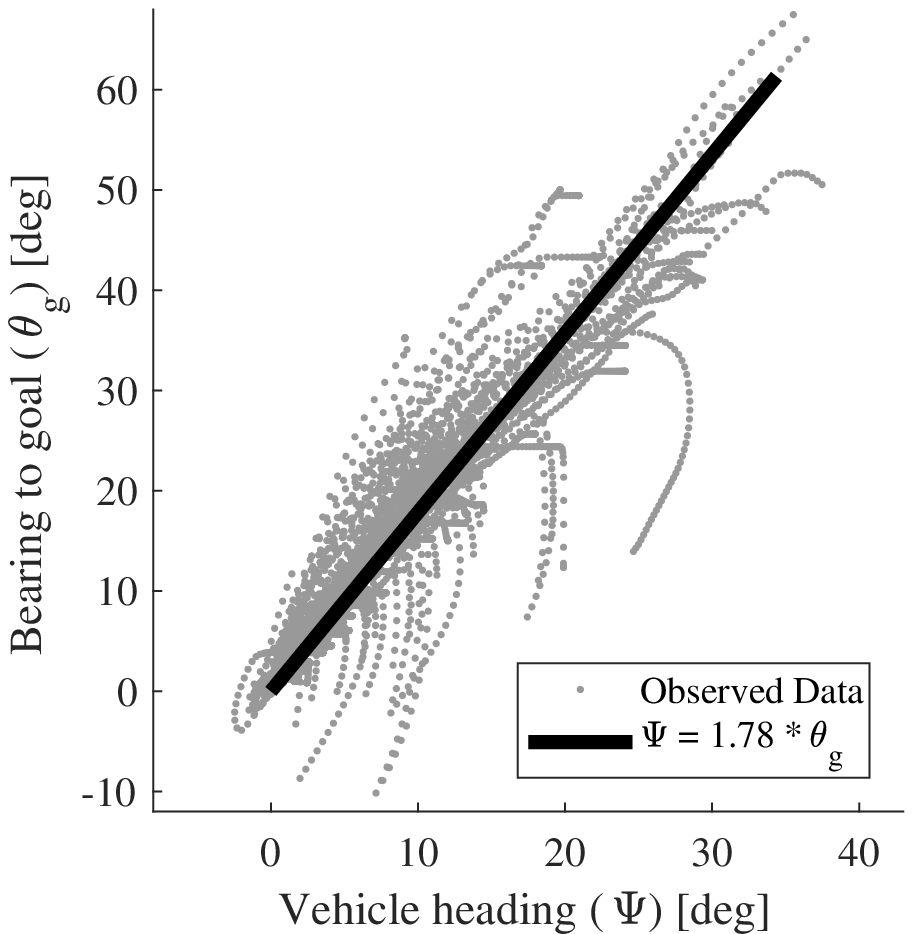}
    \label{fig:perc_guid_model}}
    \hfil
  \subfigure[Velocity vector field resulting from perceptual guidance policy.]{\includegraphics[height=4.5cm]{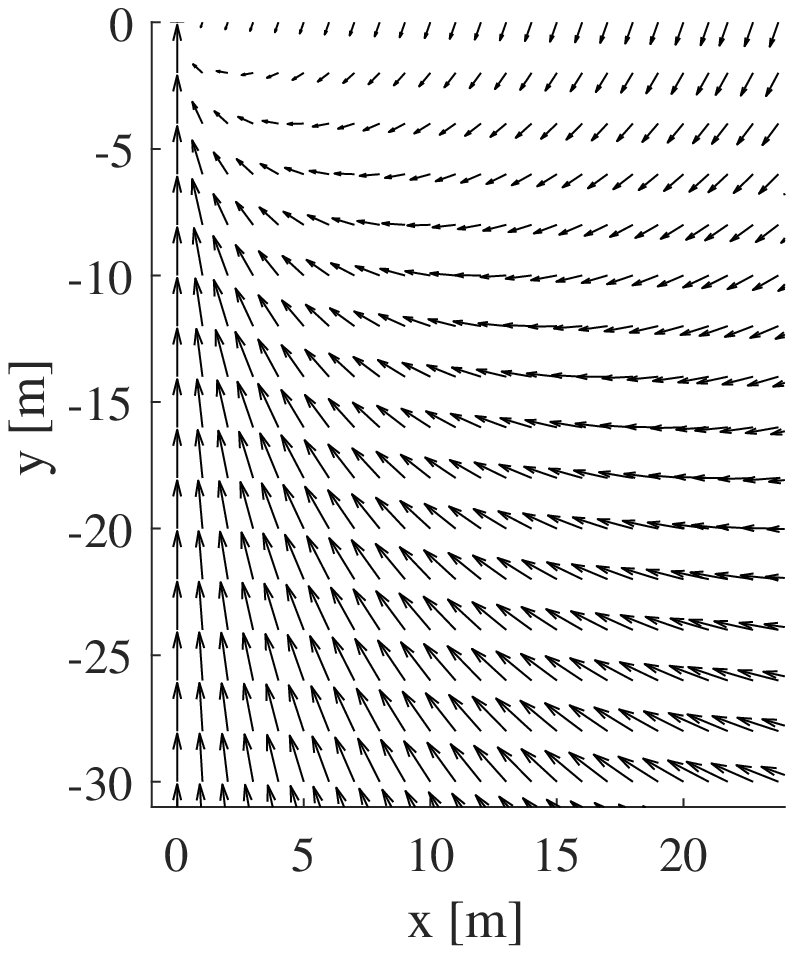}
  \label{fig:guidance_vvf}}
  \hfil
  \subfigure[Spatial cost-to-go map resulting from perceptual guidance policy.]{\includegraphics[height=4.5cm]{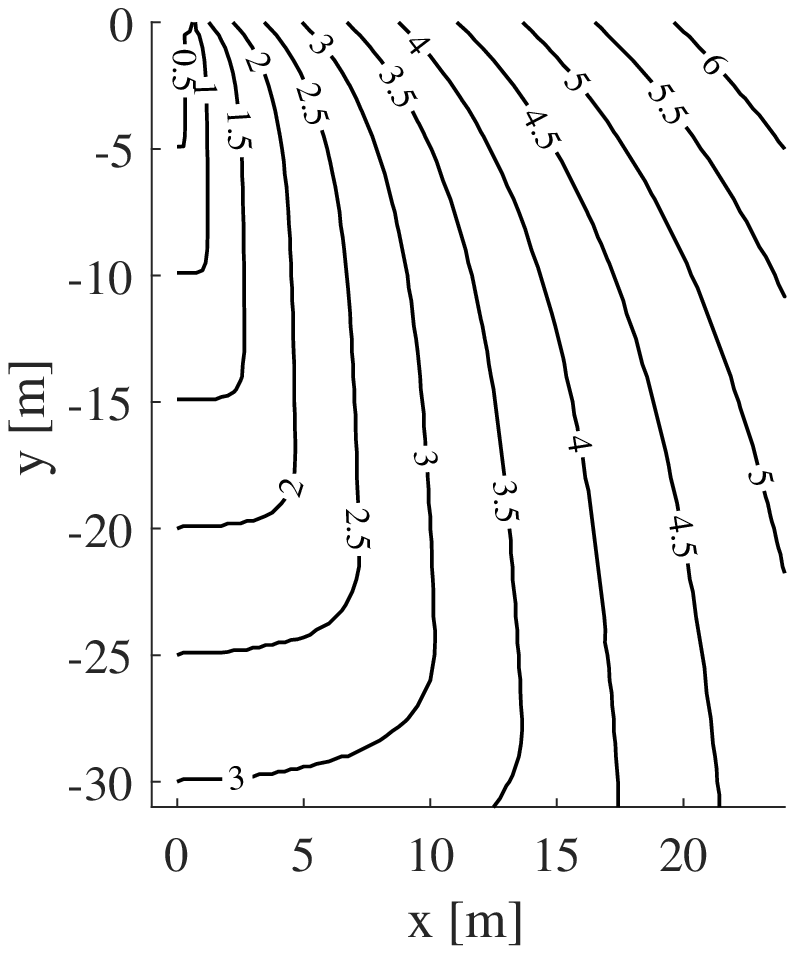}
   \label{fig:guidance_ctg}}
  \caption{Reference guidance policy.}
  \label{fig:reference_guidance_policy}
\end{figure}
During the planning process, \textbf{predictTrajectory} uses the nominal spatial guidance policy to generate reference trajectories between subgoals that satisfy the system dynamics in Eqn. \ref{eqn:vehicle_dynamics}. Rather than represent the policy as a spatial map \cite{mettler2013mapping, feit2016extraction}, the guidance function implements the policy as a feedback relationship between vehicle heading $\psi$ and bearing to subgoal $\theta_g$, as illustrated in the guidance geometry diagram, Fig. \ref{fig:guidance_geometry}.
This approach is based on the concept of perceptual guidance, suggesting that humans generate motion using simple relationships between perceived cue measurements and vehicle motion that approximate optimal behavior.
A feedback policy such as this reduces computational complexity by eliminating the need to evaluate VVF and CTG functions at each point along a trajectory.
The guidance policy is of the form $ \psi_{ref} = k * \theta_g $, where the perceptual quantity $\theta_g$ is the minimal set of relevant goal information needed to specify an action across equivalence classes. 
The VVF resulting from the feedback policy is:
\begin{eqnarray}
\mathbf{v}_{ref} = \pi_{pg}(\theta_g) &=& 
v_{mag}(\dot{\theta}_g)
\left[
\cos{k \theta_g} \; \sin{k \theta_g}
\right] ^T
\label{eqn:pg_policy}
\end{eqnarray}
In Eqn. \ref{eqn:pg_policy}, the reference velocity magnitude is based on the system lateral acceleration limit, and rate of change of vehicle heading: $v_{mag} = \max(v_{min}, a_{y,max} / k\dot{\theta_g})$. Eqn. \ref{eqn:pg_policy} acts as a set of simplified system dynamics, approximating feasible trajectories for the real vehicle system. This nominal policy is an example of an approach inspired by satisficing, since it a simplifies guidance as a sparse function of a single perceptual variable that is common across the task domain. Note that implementing this guidance policy additionally requires a tracking controller to generate control inputs to the vehicle that drive the system to follow the reference trajectory.

The nominal policy model is validated using observed human behavior recorded in prior work \cite{feit2015experimental, feit2016extraction}. Fig. \ref{fig:perc_guid_model} depicts a scatter of $\psi$ vs. $\theta_g$ recording during a simulated first-person driving task. The resulting scatter shows that a human subject's guidance behavior can be modeled as a linear policy function, approximated by Eqn. \ref{eqn:pg_policy} with $k = 1.78$, and shown by the regression line in Fig. \ref{fig:perc_guid_model}. Figs. \ref{fig:guidance_vvf} and \ref{fig:guidance_ctg} depict the resulting velocity vector function and spatial cost-to-go for this perceptual guidance policy. Spatial cost-to-go (CTG) is computed by integrating the guidance policy VVF from the initial state until the goal is reached. An important aspect of the policy is that it depicts the need to reduce speed to take sharp turns; at high angles to the goal, VVF magnitude is reduced, and CTG increases at a higher rate.

\begin{figure}
\centering
\includegraphics[height=4cm]{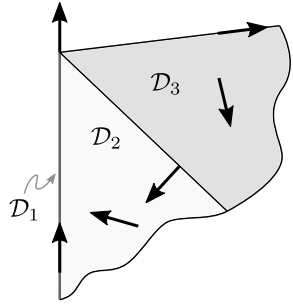}
\caption{Guidance policy stability domains.}
\label{fig:guidance_stability_domains}
\end{figure}

\subsection{Guidance Stability}

Planning-level task finite-time stability requires that a path exists consisting of a finite number of subgoals, and that trajectories between each pair of subsequent subgoals can be completed in finite time. The latter conditions requires finite time stability of the guidance policy. To show this, the perceptual guidance strategy in Eqn. \ref{eqn:pg_policy} is expressed in terms of distance to the goal and bearing error $\{d, \theta_g\}$:
\begin{eqnarray}
\left[ \dot{d} \; \dot{\theta}_g \right]^T &=& \left[ - v \cos (k \theta_g) \; -\frac{v}{d} \sin (k \theta_g) \right]^T
\end{eqnarray}

To verify finite-time stability, the guidance domain is divided into three subsets, $\mathcal{D}_1 = \left[0, \infty \right) \times 0$, $\mathcal{D}_2 = \left[0, \infty \right) \times \left(0, \pi/2k \right)$, and $\mathcal{D}_3 = \left[0, \infty \right) \times \left[ \pi/2k, \pi/k \right)$ as illustrated in Fig. \ref{fig:guidance_stability_domains}. When the system is in $\mathcal{D}_3$, the Lyapunov function $V = \theta_G$ ensures finite-time convergence into $\mathcal{D}_2$ by guaranteeing that $\dot{V} < -c$ when $\theta_G < \pi/k$. Within $\mathcal{D}_2$, the Lyupanov function $V = d + \theta_g$ is used, with derivative $\dot{V} = \frac{dV}{d \mathbf{x}_p}\dot{\mathbf{x}}_p =  - v (\cos (k \theta_g) + \frac{1}{d} \sin (k \theta_g))$.
Based on Eqn. \ref{eqn:finite_time_bounded}, the system is finite-time stable when $\dot{V} \le -c(V(\mathbf{x}_p))^{\alpha}$, resulting in:
\begin{eqnarray}
- v (\cos (k \theta_g) + \frac{1}{d} \sin (k \theta_g)) &\le& -c \left( d + \theta_g  \right)^{\alpha} \nonumber
\end{eqnarray}
Taking $\alpha = 0$, the system is finite-time stable when $c = v_{min} \min \left[ \cos (k \theta_g) + (1/d) \sin (k \theta_g) \right] > 0 $. This is conservatively satisfied when $v_{min} > 0$ and $0 < \theta_g < \pi / 2k$.
Finally, the system may begin in  $\mathcal{D}_1$, some distance from the goal but $\theta_G = 0$. In this case $V = d$ and $\dot{V} = -\dot{d} = -v_{min}$, and finite-time stability is ensured if vehicle minimum speed is lower-bounded.
Furthermore, settling time is bounded by $T(\mathbf{x}_0) \le V(\mathbf{x}_0) / c = (d + \theta_g) / c$, which is an admissible planning heuristic. Such a conservative bound allows stability to be robustly ensured over a range of guidance policies that may result from modeling errors, system failures, or environmental uncertainty.

\subsection{Computational Complexity}

Evaluating the computational complexity of SGP provides insight into the types of tasks for which SGP generates efficient results. The complexity of the A* search used in SGP depends on the heuristic quality: with no heuristic, complexity is $O(b^d)$ for average search depth $d$ and branching factor $b$ \cite{russell2003artificial}. When a heuristic is used that has log-bounded error, i.e. $\|h(\mathbf{x}) - h^*(\mathbf{x})\| \le \log h^*(\mathbf{x})$, search complexity becomes polynomial, and if the heuristic perfectly matches actual edges costs, complexity is linear, $O(d)$.

In the current application, graph edges are defined implicitly by the \textbf{findNeighbors()} function, therefore the branching factor must be determined empirically based on the number of expanded nodes and solution depth. The effective branching factor, $b^*$ is found by solving the equation:
\begin{equation}
N + 1 = 1 + b + b^2 + ... + b^d
\label{eqn:search_complexity}
\end{equation}
In Eqn. \ref{eqn:search_complexity}, $N$ is the total number of nodes expanded \cite{russell2003artificial}. 
The effective branching factor is therefore an empirical measure of fit between the heuristic and environment topology, $b^*=1$ being the ideal case when the heuristic specifies the exact actual cost of a path. In this work, we use a Euclidean distance heuristic, therefore the effective branching factor for a particular environment measures how the environment connectivity differs from linear, Euclidean paths. For example, a highly convoluted maze environment might result in a high branching factor due to the inaccuracy of approximating a path to the goal with a straight line. 

The subgoal properties defined in Sec. \ref{sec:structural_properties} provide additional limits on SGP computational complexity. 
When the guidance policy is optimal, based on Lemma \ref{lem:composite}, subgoals define the optimal composite path using the minimum-length node sequence, thereby minimizing graph search depth. In addition, the necessary conditions reduce branching factor through a-priori identification a subset of the task domain that are feasible path cost minima. These two conditions contrast with roadmap or random-sampling based planners that tend to over-discretize the environment to ensure that the graph provides sufficient coverage to obtain near-optimal solutions.

In the examples shown below, we set a constant maximum branching factor of $\epsilon$ to limit computational complexity to $O(\epsilon^d)$. This limit accounts for the conservatism of the heuristic that might otherwise result in a high branching factor, but also acts as a form of satisficing that humans may use to prune their decision tree when too many choices are available.


\section{Experimental Evaluation}
\label{sec:experiment}

\begin{figure}[h]
    \centering
    \subfigure[SGP Process Flow.]{\includegraphics[height=3cm]{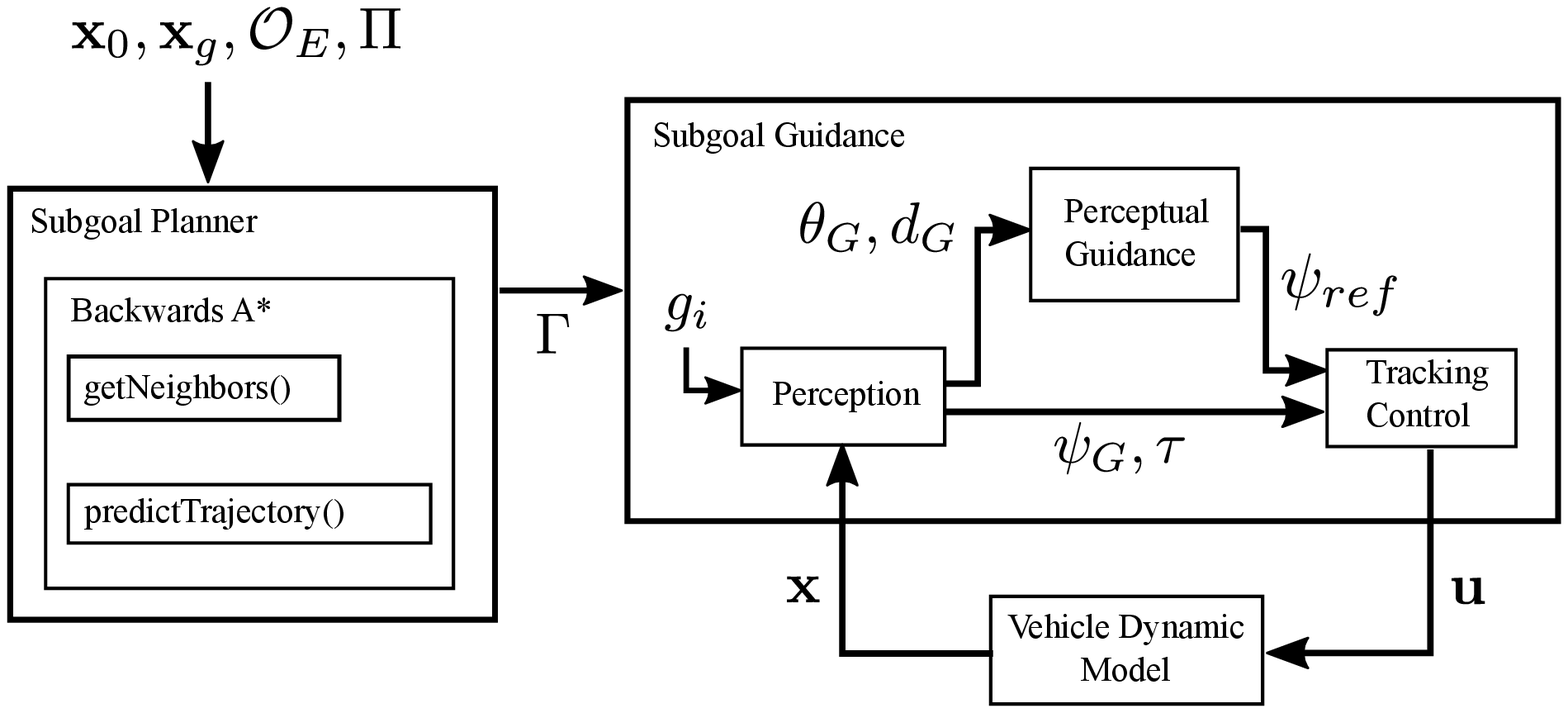}
    \label{fig:sgp_guidance}}
    \hfil
    \subfigure[RRT* Process Flow.]{\includegraphics[height=3cm]{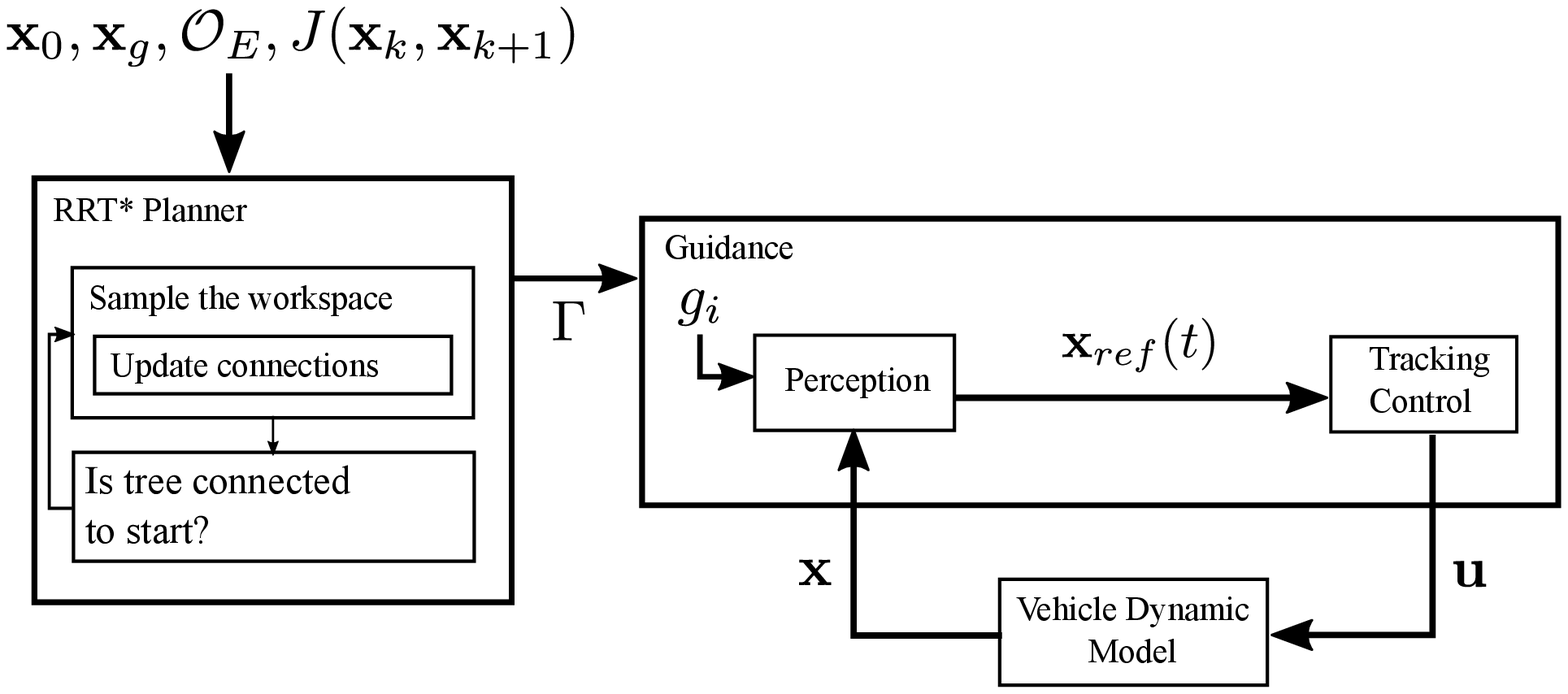}
    \label{fig:rrt_guidance}}
    \caption{RRT* and SGP Guidance Processes.}
    \label{fig:guidance_process}
\end{figure}

\subsection{Evaluation Approach}

Subgoal planning (SGP) performance and solution characteristics are evaluated by comparing results with those produced by a rapidly exploring random tree (RRT*) planner \cite{karaman2010optimal, lavalle1998rapidly}. Fig. \ref{fig:guidance_process} gives an overview of the two processes.
The RRT* planner in Fig. \ref{fig:rrt_guidance} samples locations in the task workspace, $\mathbf{p} \in \mathcal{W}$, and forms rectilinear connections that minimize the overall path cost. The cost of each RRT* edge, approximating the actual travel-time of the unicycle vehicle, is computed as the following:
\begin{equation}
J(d_G, \theta_G) = d_G/v_{max} + c \theta_G^3
\label{eqn:rrt_cost_function}
\end{equation}
In Eqn. \ref{eqn:rrt_cost_function}, for a node $\mathbf{x}_k = [ \mathbf{p}_k, \mathbf{v}_k ]$, segment length is $d_G = ||\mathbf{p}_{k+1} - \mathbf{p}_k||$, and $ \theta_G $ is the turning angle between vectors $\mathbf{v}_k$ and $\mathbf{v}_{k+1}$ (similar to \cite{feit2010travel}). 
The number of samples $k$ used by the RRT* planner can be adjusted to trade-off solution performance for computation time. This adjustment allows RRT* to be tuned to provide the best comparison with SGP. For a specific RRT* run, if the planner cannot compute a solution with the initial $ k $ samples, RRT* generates additional samples until a solution is found, so actual CPU time varies. 

To compare solution performance, trajectories are simulated by connecting a feedback controller to the unicycle model such that it tracks the reference position and heading computed by each planner.
\begin{eqnarray}
\begin{bmatrix}
\dot{u}_{lat,i} \\
\dot{u}_{lon,i} \\
u_{lat} \\
u_{lon}
\end{bmatrix}
=
\begin{bmatrix}
\psi_{err} \\
a(u_{lat}) \\
k_1 u_{lat,i} + k_2(v) \psi_{err} \\
k_3 u_{lon,i} + k_4 a(u_{lat})
\end{bmatrix}
\label{eqn:tracking_controller}
\end{eqnarray}
The controller in Eqn. \ref{eqn:tracking_controller} includes integral and proportional feedback, with control gains $k_1 = k_3 = 1.0$, $k_2(v) = 10 + 6v$, and $k_4 = 0.4$. The term $a(u_{lat}) = 2.5 (|u_{lat}| - 0.6)$ is chosen to enable the controller to reduce speed during turns.
Each planning solution is associated with both a planned cost: the sum of the planned edge costs between nodes or subgoals, and an actual cost: the time the unicycle tracker takes to reach the goal following the solution path.

\begin{figure}[tbph]
\vspace{5mm}
\subfigure[RRT* path planning result: 170 node samples, 5 trajectories at each start location.]{\includegraphics[height=4.25cm]{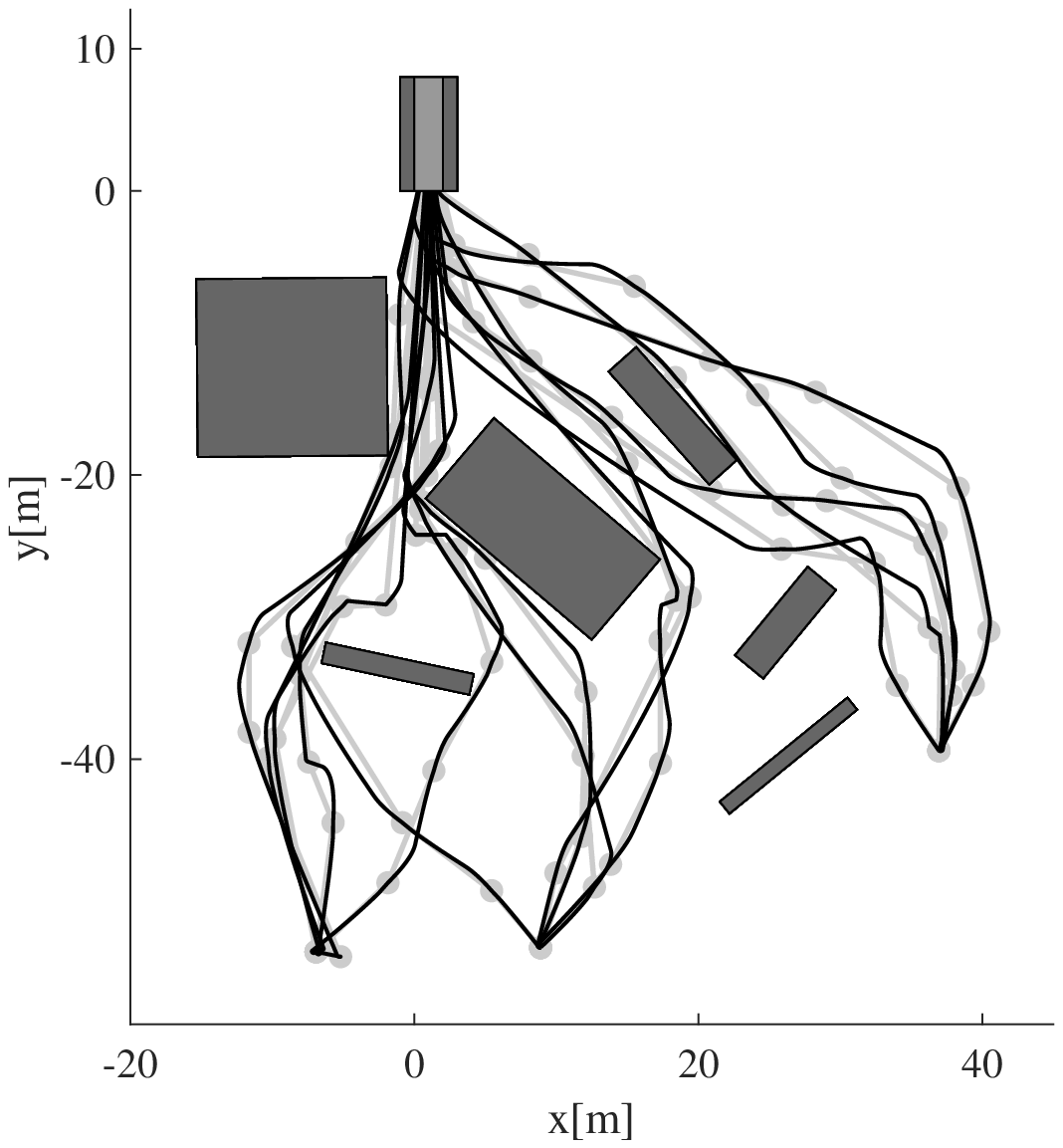}
\label{fig:rrt_example_trajectories}}
\hfil
\subfigure[SGP path planning results: $\epsilon=$ 10 sec and $n_{limit}$ = 5.]{\includegraphics[height=4.25cm]{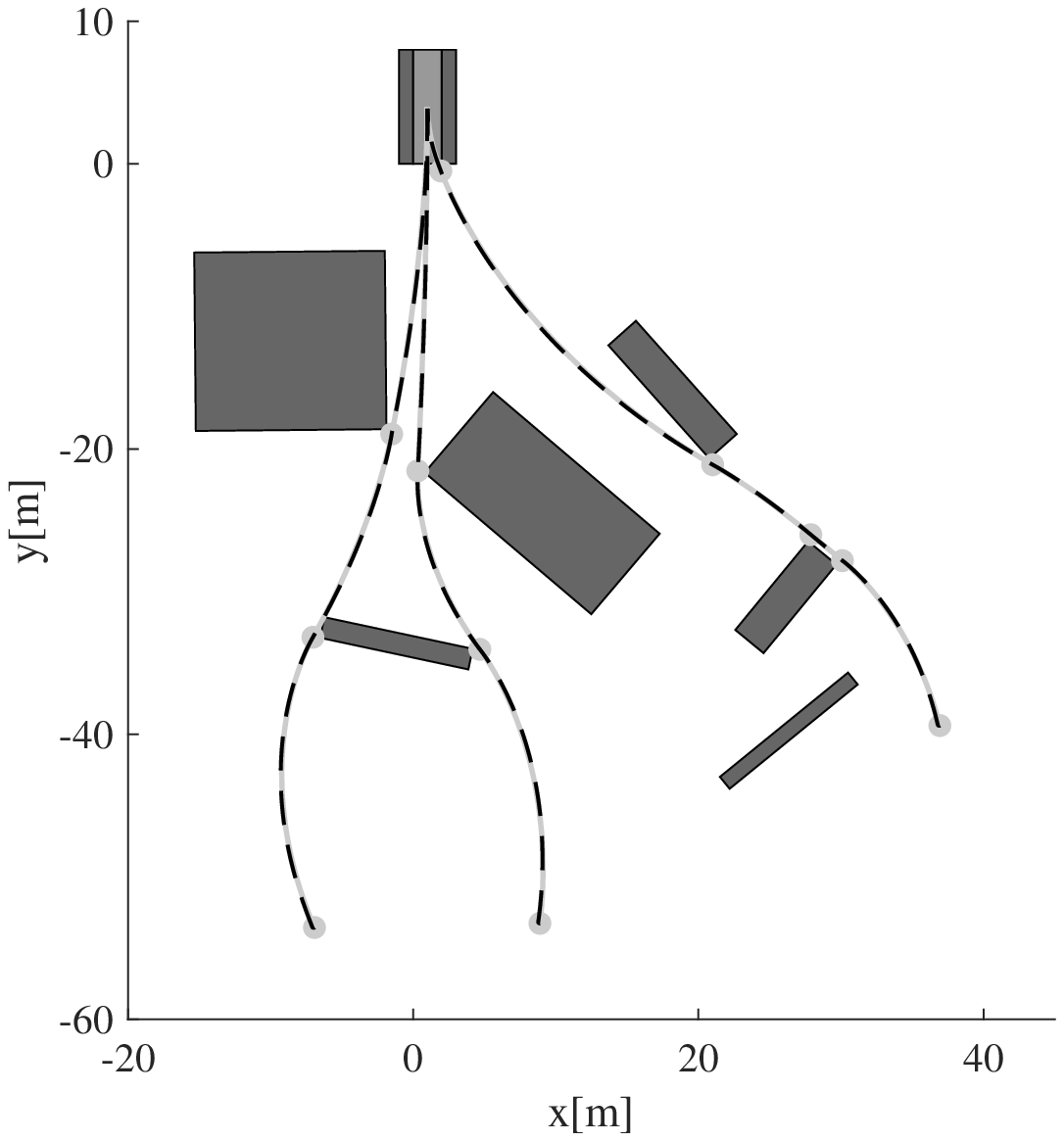}
\label{fig:sgo_example_trajectories}}
\hfil
\subfigure[RRT* cost and CPU time relative to SGP.]{\includegraphics[height=4.25cm]{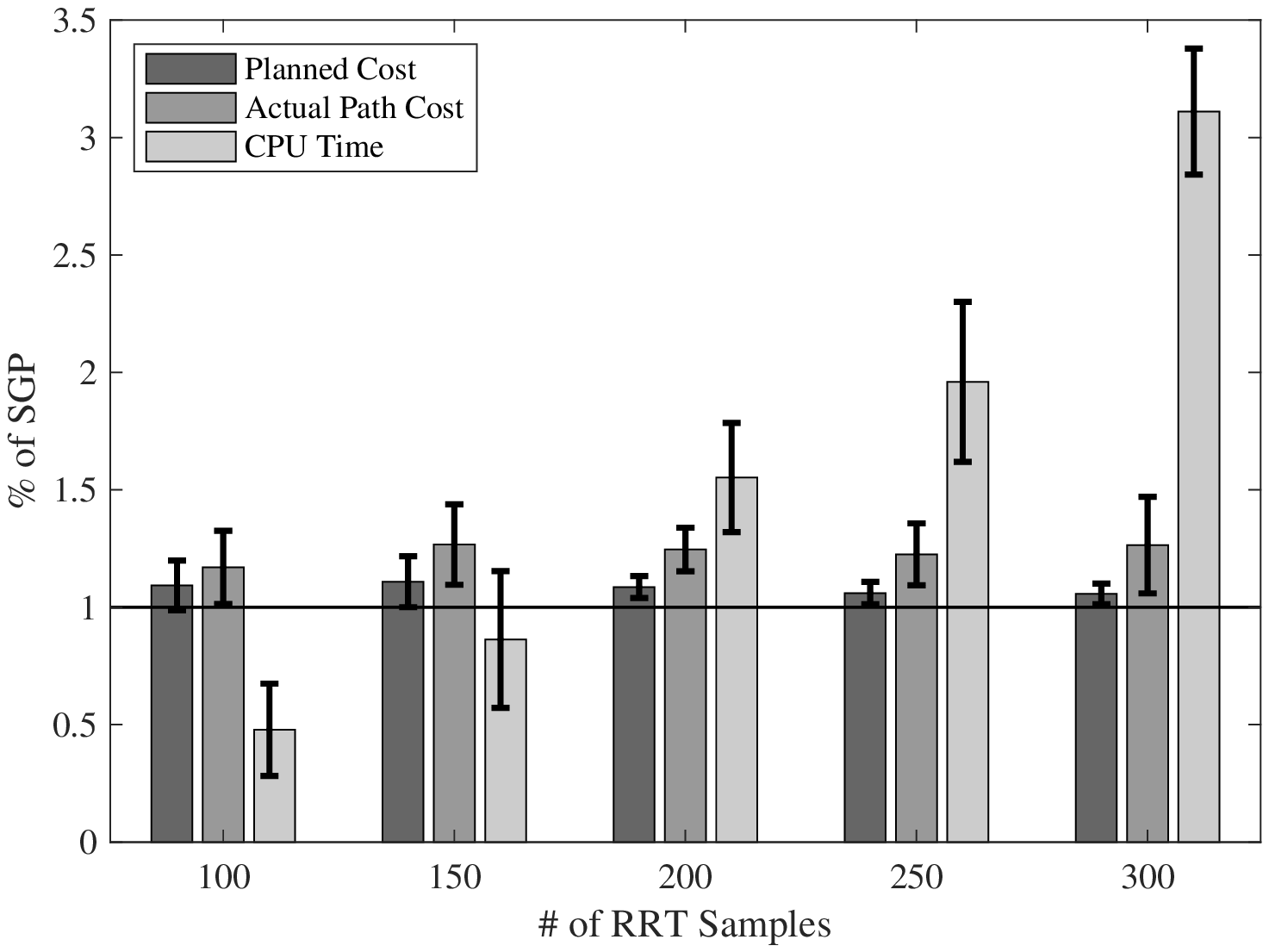}
\label{fig:performance_comparison}}
\caption{RRT* vs. SGP performance and planning time comparison for three selected starting locations.}
\label{fig:rrt_sgp_comparison_single}
\end{figure}

\subsection{Planning Performance Comparison}
Planning computation time is evaluated by generating a series of solutions from three selected starting points within the uniform-obstruction course (as used in \cite{feit2016extraction}). Fig. \ref{fig:rrt_example_trajectories} shows 15 RRT* solutions, each using 200 node samples. Fig. \ref{fig:sgo_example_trajectories} shows the SGP solutions for the same start locations. Fig. \ref{fig:performance_comparison} compares RRT* planned cost, actual path cost, and processing (CPU) time for these paths compared to SGP, over a range of node sample quantities. 

Results show that RRT* cpu time increases with sample quantity as expected, matching SGP cpu time at close to 170 samples. 
The planned cost decreases gradually with sample quantity, and is about 10\% greater (slower) than SGP for equal cpu time. Actual path cost remains at about 25\% above SGP path cost. The hypothesis is that the RRT* path actual cost remains high because the tracker generates many abrupt control actions in response to the short path segments in the RRT* solution. At 300 samples, RRT* takes over 3 times the cpu time as SGP, but average planned and actual path costs still greater than SGP.

\begin{figure}[tbh]
\vspace{5mm}
\centering
\subfigure[RRT* example trajectories.]{\includegraphics[height=5cm]{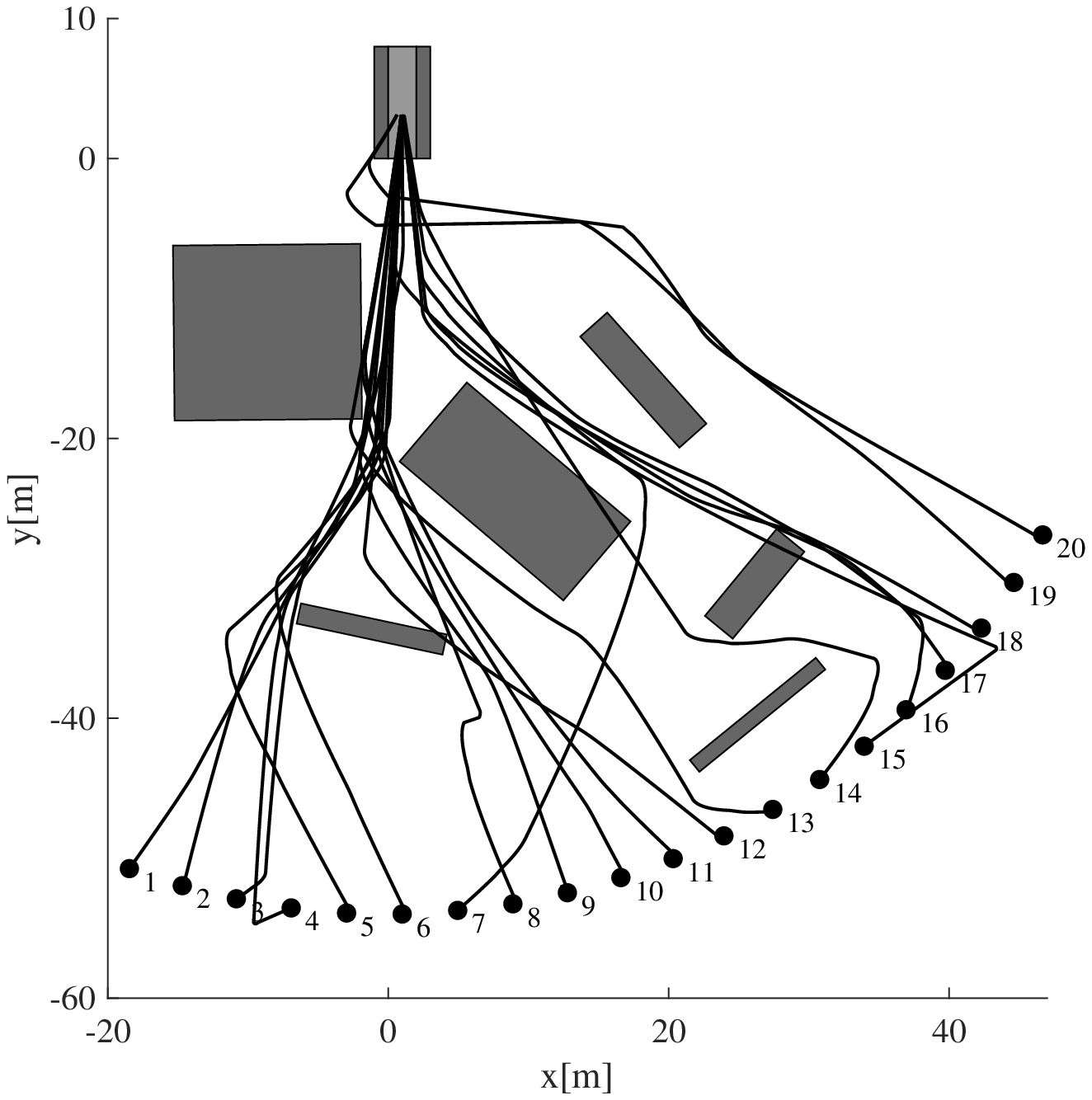}
\label{fig:rrt_example_traj}}
\hfil
\subfigure[SGP example trajectories.]{\includegraphics[height=5cm]{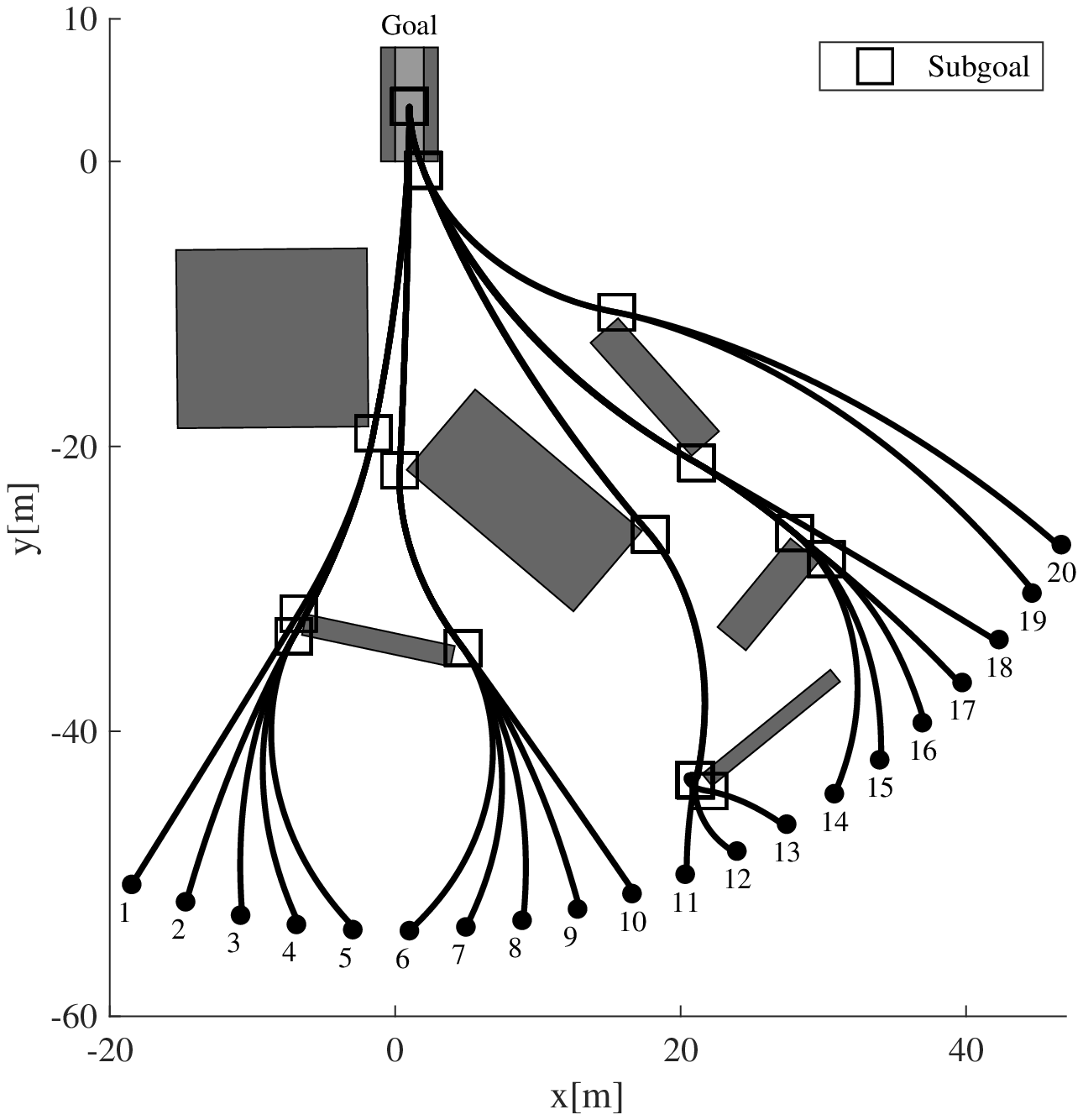}
\label{fig:sgp_example_traj}}
\caption{RRT* and SGP example trajectories.}
\label{fig:rrt_sgp_trajectories_multiple}
\end{figure}

Next, SGP and RRT* solutions are compared over the entire task domain, with paths generated from each starting location. 
Fig. \ref{fig:sgp_example_traj} shows the resulting SGP solutions. 
For RRT*, ten paths are generated at each start location, each using $k=170$ samples, and Fig. \ref{fig:rrt_example_traj} shows the lowest-cost RRT* trajectory for each. 
Fig. \ref{fig:path_time_comparison} compares the resulting costs for SGP and RRT* paths vs. start location, with RRT* exceeding SGP for all start locations but one. 
Fig. \ref{fig:run_time_comparison} plots the computation time for SGP and RRT* methods vs. start position.

\begin{figure}[tbp]
\centering
\subfigure[Actual path time: RRT* vs. SGP.]{\includegraphics[height=2.75cm]{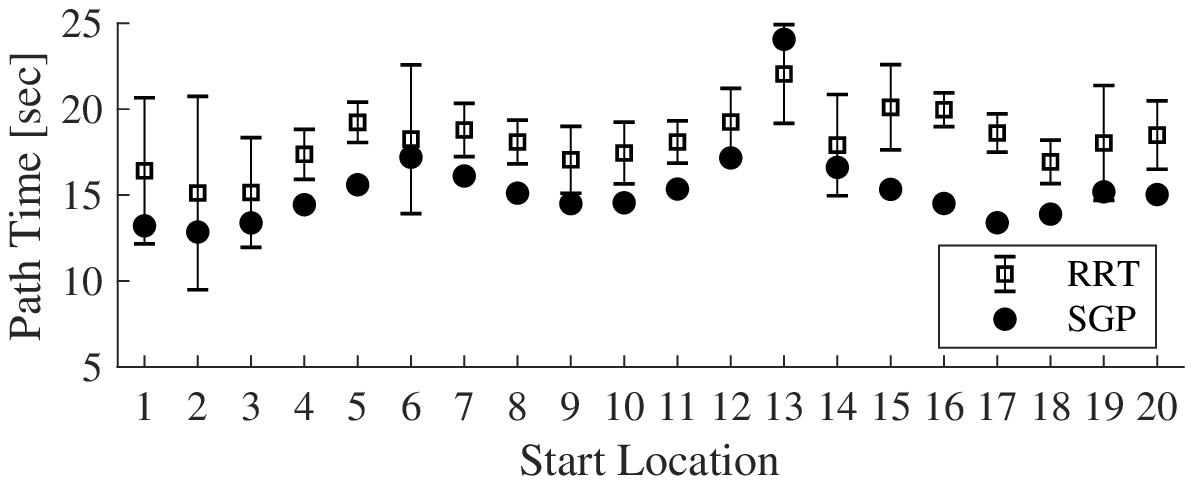}
\label{fig:path_time_comparison}}
\hfil
\subfigure[CPU time: RRT* vs. SGP.]{\includegraphics[height=2.75cm]{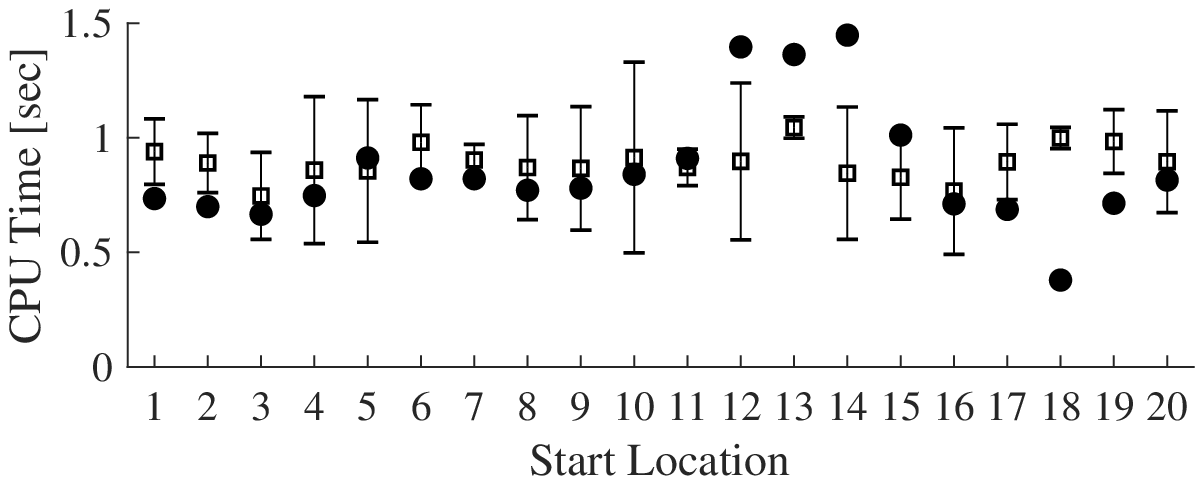}
\label{fig:run_time_comparison}}
\caption{RRT* vs. SGP performance and planning time comparison across multiple paths.}
\label{fig:rrt_sgp_comparison_multiple}
\end{figure}

\subsection{Additional Planning Examples}

\begin{figure}[tbph]
\centering
\subfigure[SGP solution trajectory and speed profile. Mean speed is 2.97 m/s, path distance is 64.0 m, and path time is 21.42 seconds. CPU time is 0.551 seconds.]{
\begin{tabular}[b]{c}
\includegraphics[height=4.5cm]{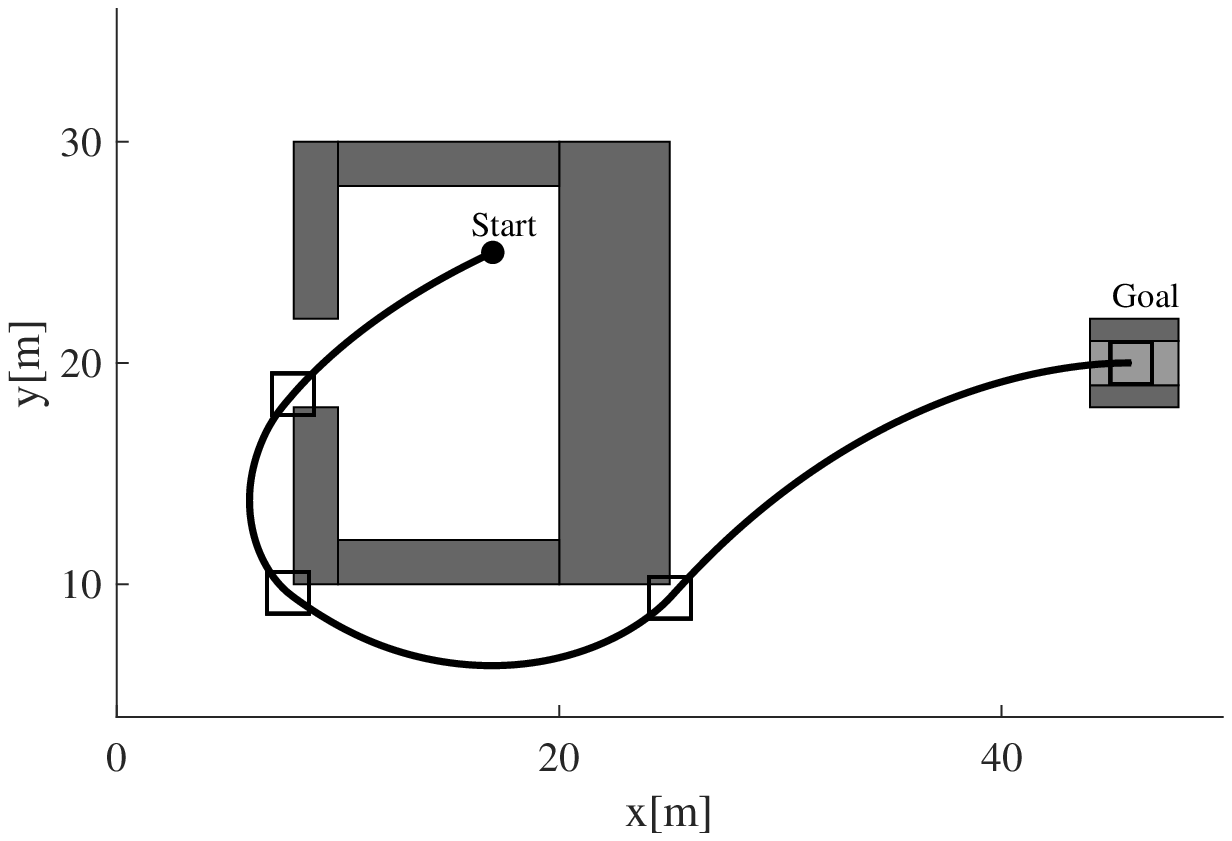} \\
\includegraphics[height=1.6cm]{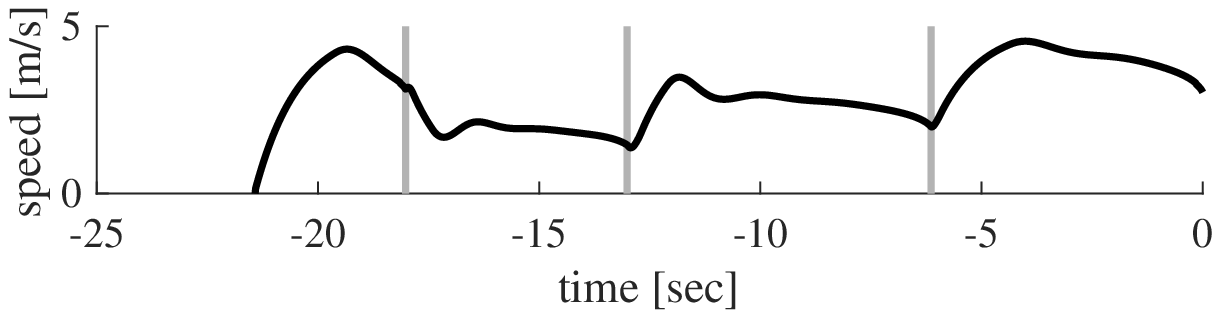}
\label{fig:sgp_u_course}
\end{tabular}
}
\hfil
\subfigure[RRT* was run three times to show a range of results. Mean speed across runs is 2.52 m/s, mean path distance is 53.3 m, and mean path time is 23.53 seconds. Mean cpu time to find a path is 4.56 seconds, using an average 375 samples.]{
\begin{tabular}[b]{c}
\includegraphics[height=4.5cm]{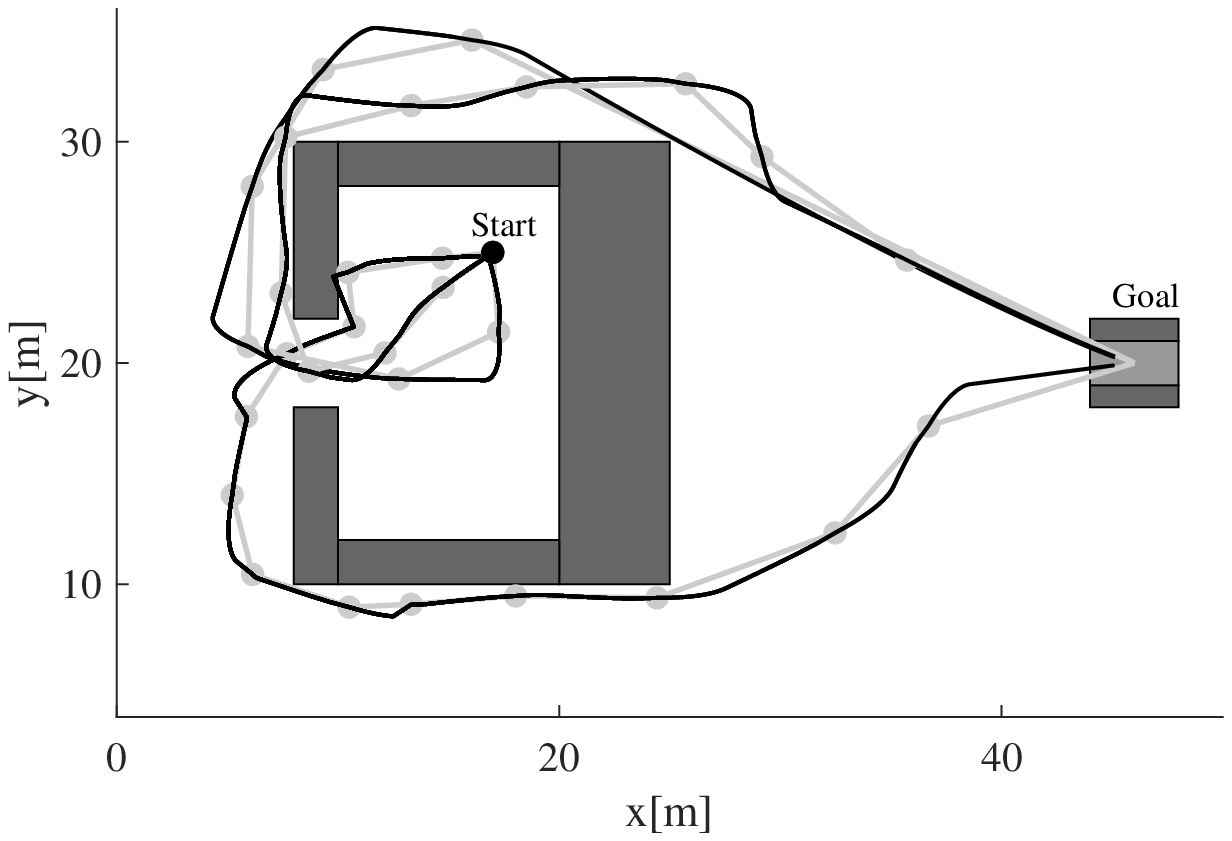} \\
\includegraphics[height=1.6cm]{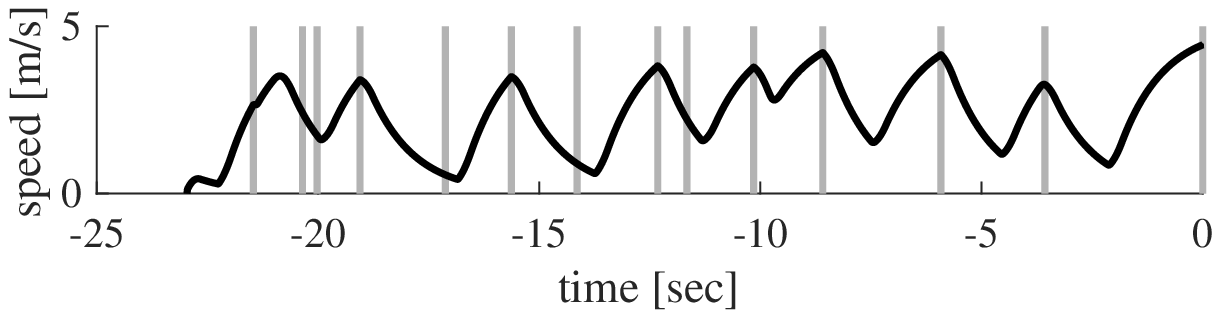}
\end{tabular}
\label{fig:rrt_u_course}}
\caption{Trap-course solutions.}
\label{fig:rrt_sgp_u_comparison}
\end{figure}

A U-shaped course is used to evaluate robust planning performance (Fig. \ref{fig:rrt_sgp_u_comparison}). RRT planners typically have difficulty planning in this type of environment because they rely on sampling points within the narrow passage. 
SGP in contrast immediately identifies subgoal locations using the necessary conditions, and provides a solution with minimal changes in vehicle speed.
The resulting RRT* path is shorter (53.3 m vs. 64.0 m), but has a lower average speed (2.52 m/s 2.97 m/s) and higher path cost (23.53 sec vs. 21.42 sec) due to the many abrupt speed changes. RRT* however used on average 375 samples in this course, requiring an average of 4.56 seconds of cpu time for planning, compared to 0.551 seconds for SGP. 




SGP solutions are evaluated on two additional courses: the two-block world and hallway world. The two-block world contains two skew-angled rectangular blocks placed to the left of the goal to show obstacle avoidance in an exterior, unenclosed environment. Velocity vector field and cost-to-go plots are shown in Fig. \ref{fig:solution_trajectories_1}, summarizing solutions over the entire domain. The hallway-world demonstrates path planning in an interior environment. The main part of the course is completely enclosed by walls formed from adjacent rectangular obstructions. Note that obstructions in this course contain concave corners. These vertices are immediately excluded as subgoal candidates since there is no valid trajectory passing through them that satisfies constraints. In the hallway world, Fig. \ref{fig:hallway_ctg} depicts subgoals placed along continuous (linear) constraint boundaries at locations tangent to the guidance policy by \textbf{getEdgeSubgoals()}.

\subsection{Effective Branching Factor}

Effective branching factor is computed for each test case from the total number of nodes expanded, $N$, and the solution path depth, $d$, as in Eqn. \ref{eqn:search_complexity}. A low branching factor indicates a good fit between planner heuristic and actual path costs given the obstacle configuration of the task. For RRT*, the number of nodes expanded is equal to the number of random samples, $k$. In the uniform course, the RRT* solution has an average effective branching factor of 3.02 vs. SGP with 3.48. In the U-shaped course, the RRT* solution average effetive branching factor is 1.63 vs. SGP with 3.16. RRT* solutions however are deeper than SGP: 5.2 vs. 3.25 in the uniform course, and 10.0 vs. 4.0 in the U-shaped course. While RRT* samples more nodes than SGP uses, RRT* maintains a low branching factor because RRT* solution paths involve more nodes than SGP solutions for the same task. This result suggests that SGP could be further improved by fine-tuning the heuristic cost function.


\begin{center}
\begin{tabular}{p{3cm}|p{2cm}|p{2cm}|p{2cm}|p{2cm}}
    \textbf{Course} & \textbf{RRT*-200 $b^*$} & \textbf{RRT*-200 $d$} & \textbf{SGP $b^*$} & \textbf{SGP $d$} \\
    \hline \hline
    Uniform (20 cases) & $3.02 \pm 1.27$ & 5.20 & $3.48 \pm 0.384$ & 3.25 \\
    U-shaped & $1.63 \pm 0.13 $ & 10.33 & 3.16 & 4.0
\end{tabular}
\end{center}

\vspace{15mm}
\begin{figure}
  \centering
  \subfigure[Two-block world velocity vector field.]{\includegraphics[height=5.5cm]{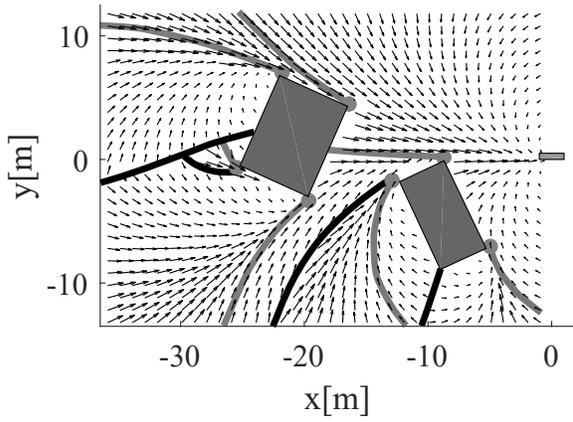}}
  \label{fig:two_block_vvf}
  \subfigure[Two-block world time-to-go map and subgoals.]{\includegraphics[height=5.5cm]{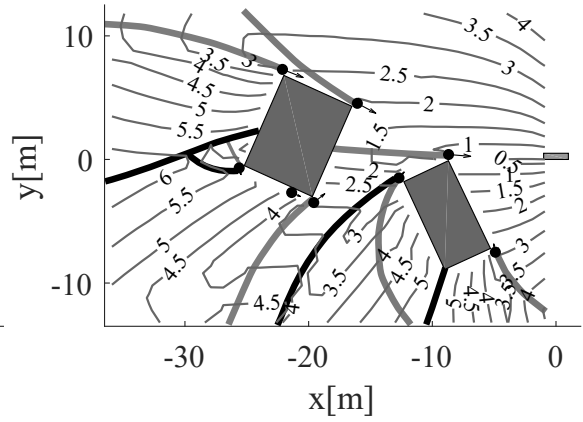}} \\
  \vspace{10mm}
  \subfigure[Hallway world velocity vector field.]{\includegraphics[height=7cm]{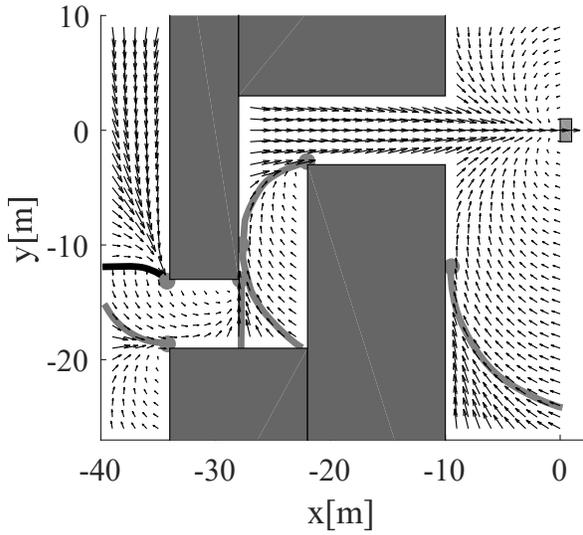}}
  \subfigure[Hallway world time-to-go map and subgoals.]{\includegraphics[height=7cm]{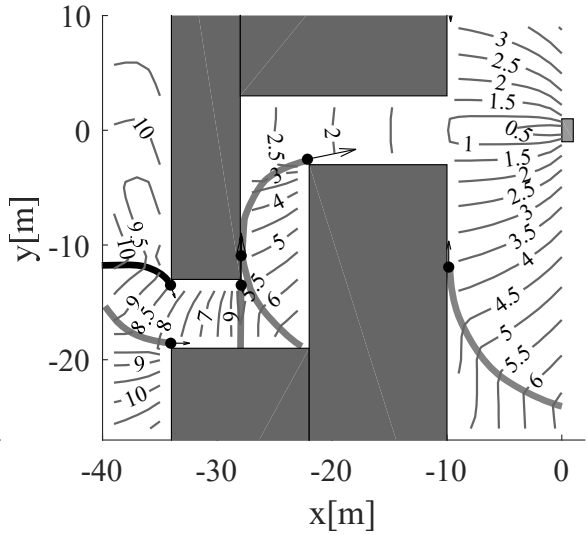}
  \label{fig:hallway_ctg}
  }
 \caption{Subgoal planning test solutions. Wide gray lines indicate bounding trajectories. Black lines indicate partition separatrices.}
 \label{fig:solution_trajectories_1}
\end{figure}

\section{Discussion}
\label{sec:discussion}
This section discusses subgoal planning implementation issues involved with dynamic tasks, higher-dimension configuration spaces, and sensory process integration.

\subsection{Dynamic Planning Tasks}
The subgoal planning approach is most applicable to tasks involving a high level of dynamic maneuvering. Verma and Mettler \cite{verma2016scaling} quantify this relationship between system maneuver capability and environment geometry by introducing the manuever-length scale ratio (MLSR). Vehicles with a high maneuverability relative to distances between obstacles (e.g. passenger aircraft cross-country routing) can adequately plan routes in the spatial domain using visibility graph or way-point planning. Agile vehicles operating at high speed and close to obstructions, such as UAVs flying in an urban environment, must account for vehicle dynamic constraints at the trajectory planning level. The subgoal guidance algorithm provides a method, using IPs as the units of behavior, to identify dynamically optimized trajectories more efficiently than by using full-state, sampling-based planning approaches. The subgoal planning approach can be extended to any robotic, dynamic planning task, providing a more efficient link between discrete task and continuous motion planning.

\subsection{Three-Dimensional Spatial Domain}
The present SGP motion planning implementation has been formulated for a 2D spatial domain ($\mathcal{W} \in \mathbb{R}^2 $). In $\mathbb{R}^2$, necessary conditions specify a discrete set of subgoal candidate points, i.e obstacle vertices. A task in a higher-dimension configuration space can result in a subgoal candidate set containing continuous subsets. For example, in a navigation problem in $\mathcal{W} \in \mathbb{R}^3 $ with a single spherical obstacle, the set of subgoal candidates consists of a continuous circular manifold of points where the optimal velocity vector is tangent to the sphere surface. In this case, the planner must choose one or more discrete subgoal candidates from each continuous subset while exploring neighbor nodes. Determining optimal subgoal candidates in this case involves higher computational cost. Nevertheless, a satisficing approach may be applied that picks a sub-optimal subgoal state. Future work is needed to investigate subgoal planning for higher-dimension spaces in more detail.

\subsection{Uncertain Environments}
A fundamental issue in uncertain or unknown environments, and the motivation for dynamic programs in general, is that planning information flows from the goal state towards the start; each subgoal state depends on the next subgoal, but is the reverse direction in which the agent experiences the environment. 
This disparity causes two primary issues, first, the agent must make planning decisions with limited subgoal information. 
Information both perceived during prior runs and extrapolated from the environment provide clues to determine the best actions. 
Verma et al \cite{verma2016computational} investigate environment perception, representation and learning in human motion guidance experiments. 
Results show that subjects improve planning over repeated trials by using knowledge about future parts of the task from prior trials. 
In addition, subjects make decisions using meta-information about how they expect a space is connected, for example, the subjects didn't explore or focus their attention towards shortcuts in cases where they assumed the path to be blocked based on their prior expectation of the environment layout. 
Another issue is that when information is limited, incorrect decisions are likely. To prevent catastrophic outcomes, an agent must take into account information constraints in the planning process. Accounting for these constraints may include choosing safe, known routes over potentially faster routes that are less certain. Subgoals, and the associated partitioning of the problem space, provide a method for propagating these uncertainties through the task environment.

In uncertain environments, subgoal costs are random variables, characterized by a distribution $p(J(g_k))$. A utility function, $U : J(g_k)) \rightarrow \mathbb{R}^+$, defines a positive value to the agent of a task state based on the cost-to-go. The resulting decision policy is expressed as the maximization of expected utility of:
\begin{eqnarray}
g_{k+1} = \gamma(g_k, e_k) = \argmax_{g_{k+1} \in G} E\left[ U(p(J(g_{k+1}))) \right]
\label{eqn:planning_unc_policy}
\end{eqnarray}
Stability, robustness and performance are characterized by the utility distribution of each subgoal. The utility function may be designed to include nonlinearities that emphasize relevant decision making characteristics such as diminishing return of low cost paths, or avoidance of extremely high path costs \cite{russell2003artificial, bernoulli1954exposition}. If the chosen utility function $U(\cdot)$ is strictly monotonically decreasing with cost-to-go, then the resulting finite-time stability condition is that the utility gain from each subgoal must be bounded from below by a constant $\epsilon > 0$:
\begin{equation}
E\left[ U(J(g_k) - J(g_{k-1})) \right] \ge \epsilon
\label{eqn:finite_stability_unc}
\end{equation}
Eqn. \ref{eqn:finite_stability_unc} defines an information constraint on a subgoal utility distribution that ensures stability. If no known subgoal candidates satisfy this constraint, the agent may generate new subgoal candidates that guarantee a minimum utility. 
For example, a subgoal may be added that provides an option of stopping or moving into a loitering pattern before a possible obstacle collision occurs. 
This safety guarantee would be similar to the approach used by Schouwenaars et al. \cite{schouwenaars2004receding} with a mixed-integer path planning optimization.

If a task will be repeated multiple times, the agent must balance the maximization of performance on the current trial based on available information (exploitation), with gaining task information (exploration) to improve future performance. One way to quantify this tradeoff is using decision entropy, or informational regret \cite{tishby2011information}, which can be computed across the likelihoods that each subgoal is optimal, $H_{dec} = \sum_{g_i \in G} p_{min}(g_i|\mathbf{x}) \log p_{min}(g_i|\mathbf{}x|)$, where $p_{min}(g_i|\mathbf{x})$ is the probability that choosing subgoal $g_i$ minimizes path cost, conditioned on the current agent state $\mathbf{x}$. Decision information can also be quantified by empowerment \cite{klyubin2005empowerment}, which is the maximum mutual information between actions and perception for a specific action policy. Exploration can be modeled as maximizing future empowerment; an agent should explore by choosing routes that will provide the maximum decrease in decision entropy for future runs.

\subsection{Sensory Process Integration}
During navigation tasks in the real world, agents combine environment information perceived during task execution with prior knowledge. The agent must decide on a subgoal sequence based on estimates of subgoal locations and perception of constraint boundaries. To alleviate the computational complexity of continual environment perception, receding-horizon approaches \cite{mettler2008receding} combine near and long-term planning. In addition, receding horizon planning provides a model for exploratory vs. exploitative behavior when a task is repeated for multiple trials \cite{verma2016computational}. 
In previous work, the environment was modeled by a cell-grid defining obstacle occupancy and cost over the task domain. As the agent moves through and perceives the environment, occupancy and cost cells are updated, providing information to improve future planning trials. This approach however requires a large amount of memory to keep a high-resolution map of the environment.

Subgoal guidance offers an alternative efficient hierarchical representation of the global cost-to-go information, in the form of subgoal nodes connected by guidance primitive elements. 
For humans, perception of the environment focuses on identifying subgoal candidates (i.e. obstacle corners), and verifying that projected guidance trajectories satisfy environment constraints. 
The information from this perceptual process is used to update the agent's knowledge of subgoal locations, cost distributions, and feasible connections \cite{verma2015investigation}.  
Because of this focus on only feasible routes, a sparse subgoal representation is more efficient across different environment scales than a grid representation.
Subgoals occur at the resolution needed to define the dynamic and perceptual interactions between the agent and environment, and therefore contain the minimum required information needed to perform the task.
Future work will test the use of subgoals as an environment representation for planning and learning implementations in large, uncertain task domains.

\section{Conclusions}
\label{sec:conclusion}

This paper presents an optimal control formulation of subgoal guidance strategies inspired by human guidance behavior. Despite polynomial time complexity, The subgoal planner (SGP) generates lower-cost solutions more quickly than a reference sampling-based planner in the example task domains presented above. SGP is able to reduce complexity by choosing subgoal locations based on environment constraints to avoid oversampling. In addition, unlike previous roadmap methods, SGP places subgoals based on the relationship between constraints and vehicle dynamics so as to generate dynamically optimal solutions. Stability conditions for deterministic planning are presented, and extended to general conditions for robustness and performance for a stochastic planning process.

The subgoal planning algorithm achieves its performance by exploiting structural properties of the spatial navigation and guidance task, consisting of equivalence classes, to efficiently generate optimal control solutions. The resulting algorithm generates solutions that mimic human behavior, and achieve improved performance and robustness over existing approaches. Subgoal planning provides a general method of discretizing continuous environments based on system dynamics and environment constraints.  SGP identifies discrete minimal-cost subgoal candidates in a continuous task domain using conditions based on constraint topology. The A* graph-search algorithm then computes an optimal subgoal sequence among subgoal candidates.

The subgoal necessary conditions presented here, along with the concepts of partitions and guidance elements, are examples of principles used to exploit task structure \cite{braun2010structure} 
Taken together, these primitive elements form a language of spatial behavior that can be applied to human behavior analysis and autonomous guidance algorithms.

Future work is needed investigate the extension of this approach to systems with higher-dimensional configuration spaces, uncertain environments, and to the sensory processes involved in real-time motion planning tasks. Finally, future work will investigate the application of subgoal planning concepts to provide efficient autonomous guidance solutions and intuitive human-machine interactive systems.

\section*{Acknowledgments}
This work is financially supported by the U.S. Office of Naval Research (2013-16, \#11361538) and the National Science Foundation (CAREER 2013-18 CMMI-1254906).

\bibliographystyle{new-aiaa}
\bibliography{human_guidance_references}

\end{document}